\documentclass[lettersize,journal]{IEEEtran}
\pdfoutput=1
\usepackage{amsmath,amsfonts}
\usepackage{array}
\usepackage{amsthm}
\usepackage{textcomp}
\usepackage{stfloats}
\usepackage{url}
\usepackage{mathtools}
\usepackage{tabularx}
\usepackage{booktabs}
\usepackage{pgfplots}
\usepackage{cleveref}
\usepackage{adjustbox}
\usepackage{verbatim}
\usepackage{graphicx}
\usepackage[linesnumbered,ruled,vlined]{algorithm2e}
\usepackage{cite}
\usepackage{subfigure}
\usepackage{multirow} 
\usepackage{threeparttable}
\newtheorem{definition}{Definition}
\newtheorem{theorem}{Theorem}

\begin{document}

\title{Epass: Efficient and Privacy-Preserving Asynchronous Payment on Blockchain}

\author{Weijie~Wang, 
Jinwen Liang,~\IEEEmembership{Member,~IEEE,} Chuan~Zhang,~\IEEEmembership{Member,~IEEE,} Ximeng Liu,~\IEEEmembership{Senior Member,~IEEE,}  Liehuang~Zhu,~\IEEEmembership{Senior Member,~IEEE}, and 
Song~Guo,~\IEEEmembership{Fellow,~IEEE
}
\IEEEcompsocitemizethanks{
\IEEEcompsocthanksitem Weijie Wang is with the School of Computer Science, Beijing Institute of Technology, Beijing, China. E-mail: weijiew@bit.edu.cn.
\IEEEcompsocthanksitem Jinwen Liang and Song Guo are with the Department of Computing, The Hong Kong Polytechnic University, Hong Kong. E-mail: \{jinwen.liang, song.guo\}@polyu.edu.hk. 
\IEEEcompsocthanksitem Chuan Zhang and Liehuang Zhu are with the School of Cyberspace Science and Technology, Beijing Institute of Technology, Beijing, China. E-mail: \{chuanz, liehuangz\}@bit.edu.cn.
\IEEEcompsocthanksitem Ximeng Liu is with the College of Mathematics and Computer Science, Fuzhou University, and Fujian Provincial Key Laboratory of Information Security of Network Systems, Fuzhou, China. E-mail: snbnix@gmail.com. 
\IEEEcompsocthanksitem Weijie Wang and Jinwen Liang contribute to the work equally and should be regarded as co-first authors. 
\IEEEcompsocthanksitem Chuan Zhang is the corresponding author.}
\thanks{This research is financially supported by the ``National Key R\&D Program of China" (2021YFB2700500, 2021YFB2700503), the National Natural Science Foundation of China (Grant Nos. 6223000240, 62202051), the China Postdoctoral Science Foundation (Grant Nos. 2021M700435, and 2021TQ0042), the Shandong Provincial Key Research and Development Program (Grant No. 2021CXGC010106), and the Key-Area Research and Development Program of Guangdong Province under grant No.02021B0101400003.}
\thanks{Manuscript received April 19, 2021; revised August 16, 2021.}}


\maketitle

\begin{abstract}
Buy Now Pay Later (BNPL) is a rapidly proliferating e-commerce model, offering consumers to get the product immediately and defer payments. Meanwhile, emerging blockchain technologies endow BNPL platforms with digital currency transactions, allowing BNPL platforms to integrate with digital wallets. However, the transparency of transactions causes critical privacy concerns because malicious participants may derive consumers’ financial statuses from on-chain asynchronous payments. Furthermore, the newly created transactions for deferred payments introduce additional time overheads, which weaken the scalability of BNPL services. To address these issues, we propose an efficient and privacy-preserving blockchain-based asynchronous payment scheme (\textsf{Epass}), which has promising scalability while protecting the privacy of on-chain consumer transactions. Specifically, \textsf{Epass} leverages locally verifiable signatures to guarantee the privacy of consumer transactions against malicious acts. Then, a privacy-preserving asynchronous payment scheme can be further constructed by leveraging time-release encryption to control trapdoors of redactable blockchain, reducing time overheads by modifying transactions for deferred payment. We give formal definitions and security models, generic structures, and formal proofs for \textsf{Epass}. Extensive comparisons and experimental analysis show that \textsf{Epass} achieves KB-level communication costs, and reduces time overhead by more than four times in comparisons with locally verifiable signatures and Go-Ethereum private test networks.
\end{abstract}


\section{Introduction}
Buy Now Pay Later (BNPL) allows consumers to pay in installments, usually several equal parts, and get the product immediately \cite{bnpl-crypto}. This service enables consumers to enjoy the products they need early even if they don't have enough money available, helping people to enhance their life experiences. Many shopping platforms, such as Amazon, Taobao, Jingdong, etc., support BNPL services using USD, RMB, etc. At the same time, the rise of blockchain has developed a new form of currency: digital currencies such as bitcoin and ether. According to statistics, as of early February 2023, users were paying a total of \$4.4 million per day for Ethereum transactions \cite{CryptoFees}. Many businesses are also starting to support using digital currencies to shop for products. For example, Klarna allows customers to buy digital assets \cite{Klarna}, and Affirm added cryptocurrency to allow customers to buy and sell digital currency on their application \cite{Affirm}. Therefore, it is not surprising to see that the BNPL platforms enable their customers to use digital currencies. 

Existing works such as Pay Later Project (PLP) \cite{Paylaterproject2022}, Atpay \cite{Atpay}, and Apenow \cite{Apenow}, are performed by smart contracts to provide deferred payment services to consumers. Smart contracts allow for internal constraints requiring consumers to pay BNPL service providers at specific times \cite{sinha2021blockchain}. However, the data transparency on-chain leads to serious privacy concerns \cite{shah2019blockchain,zhang2019tppr}. Transactions provided by consumers to smart contracts are visible to all nodes in the blockchain network, which also results in deferred payment transactions being as public as regular transactions. This is unacceptable in real-world e-commerce applications, as malicious third parties may analyze consumers' financial situation based on their deferred payment transactions~\cite{li2022blockchain}. In addition, it undermines the scalability of BNPL services due to the additional time overhead incurred by consumers' newly created deferred payment transactions. Each deferred payment transaction generated by a consumer appears as a new transaction in the blockchain network, and these newly generated transactions must again go through a series of processes, including acceptance, mining, propagation, and node network validation \cite{gupta2021blockchain}. This procedure requires significant processing power and time and burdens the blockchain, especially for instalment payments.

Our goal is to build a services computing asynchronous payment system with the following features. First, it protects the privacy of users' transactions so that malicious third parties do not misuse their deferred payment transaction information. Second, it has a reasonable performance overhead and scalability as the number of users increases. An essential foundation for existing work to implement deferred payments is smart contracts, which monitor the behaviour of consumers and BNPL service providers through internal constraints. However, since smart contracts are public, all participants in the blockchain network can derive the input, execution process, and output of a smart contract~\cite{Zhang2022TMC}. Therefore, smart contracts will no longer meet our requirements because they are not suitable for protecting the privacy of users' transactions. In this work, we propose a deferred payment paradigm based on public key encryption with locally verifiable signatures to guarantee users' privacy. In addition, to preserve the scalability, we use the chameleon hash trap to modify the deferred payment transaction without creating additional new transactions, thus improving the efficiency of the blockchain. However, it is still challenging to enable BNPL service providers to make changes to transactions at specific times without using smart contracts. To address this issue, we combine timed-release encryption and a redactable blockchain to enable delayed payments. Specifically, we deploy servers that provide keys so that providers offering deferred payment services can only modify transactions at specific time intervals.

The main contributions of this paper.
\begin{itemize}
    \item We propose a blockchain-based deferred payment scheme (\textsf{Epass}) that aims to address the privacy and efficiency issues in blockchain-authorized BNPL services. It reduces the on-chain burden while protecting user privacy.
    \item \textsf{Epass} has less time overhead compared to existing schemes. We deploy servers to provide time instant key so that the deferred payment service provider can only rewrite the on-chain transactions at a specific instant of time. Subsequent changes to the transaction through the trapdoor of chameleon hash only require node network verification, saving time and arithmetic power compared to generating a new transaction.
    \item \textsf{Epass} supports the protection of user privacy in the case of multiple transactions generated by users in the system. All transaction signatures of a user are aggregated into a single aggregated signature that is used for aggregation verification. We extend the single verification of locally verifiable signatures to subset verification. Later, without the service provider knowing all of the user's transactions, the subset of transactions requiring deferred payment can be decompressed from this aggregated signature.
    \item We have conducted extensive experiments on \textsf{Epass}. The results show that \textsf{Epass} has practical security features and acceptable communication cost (KB-level). Compared to the baseline, the time overhead was reduced by more than four times.
\end{itemize}

The remainder of this paper is organized as follows. Section II describes the symbols used and the cryptographic primitives involved. Section III describes the system architecture, threat model and security model. Section IV describes our proposed system in detail and gives a formal definition.  Section V analyzes the correctness and security properties and performs a performance analysis in Section VI. Section VII describes the current issues in blockchain e-commerce solutions and gives a technical overview. Finally, we conclude our work in Section VIII.

\section{Preliminary}\label{preliminary}
In this section, we first define the notations to be used. Then we describe digital signatures, chameleon hashes and timed-release encryption. 
\subsection{Notation}
The notations and corresponding descriptions in this paper are provided in \Cref{tab:addlabel}.

\begin{table}[htbp]
    \centering
    \caption{Notations}
        \begin{tabular}{ll}
            \toprule
            \textit{\textbf{Notation}} & \textit{\textbf{Description}} \\
            \midrule
            $\mathbb{G}, \mathbb{G}_T$     & bilinear groups\\
            $H,H_1, H_2$                   & crytographic hash functions \\
            $\textit{mpk}, \textit{msk}$                     & master key pair \\
            $\textit{pk}_u, \textit{pk}_m, \textit{pk}_S$             & public key \\
            $\textit{pk}_{\textit{local}}$                   & local verification key \\
            $\textit{sk}_u, \textit{sk}_m, \textit{sk}_S$             & secret key \\
            $\sigma$                       & signature \\
            $C$                            & ciphertext \\
            $\hat{\sigma}$                 & aggregated signature \\
            $\textit{tx}$                           & transaction \\
            $\textit{aux}$                    & auxiliaty information \\
            \bottomrule
        \end{tabular}%
    \label{tab:addlabel}%
\end{table}%

\subsection{Digital Signatures}
\textbf{Digital Signatures:} A digital signature is a method of identifying digital information through public key cryptography to authenticate and confirm the identity and eligibility of the signer. Two complementary operations are typically defined by digital signatures, one for signing and the other for verifying. Sender $\mathcal{A}$ signs the data to be transmitted by its own private key to generate a digest, and then transmits the digest generated by the signature and the data to be transmitted together to receiver $\mathcal{B}$. Receiver $\mathcal{B}$ verifies the signature by $\mathcal{A}$'s public key after receiving the data. Four algorithms make up a digital signature $\mathsf{DS}$ \cite{kaur2012digital, merkle1988digital} with message space $\mathbb{M}$: $ \{ \mathsf{Setup}, \mathsf{KeyGen}, \mathsf{Sign}, \mathsf{Verify}  \} $.

\begin{itemize}
     \item \textit{$\mathsf{DS.Setup} \left ( 1^{\lambda }  \right )$  $\longrightarrow \left ( \textit{pp} \right ):$} Security parameter $\lambda \in \mathbb{N}$ is taken as an input, and a public parameter $\textit{pp}$ is output. Other algorithms take the public parameter $\textit{pp}$ as an implicit input.
 
    \item \textit{$\mathsf{DS.KeyGen} \left ( \textit{pp} \right )$  $\longrightarrow \left ( \textit{sk}, \textit{pk} \right ):$} Following the entry of a public parameter $\textit{pp}$, then provide a secret key $\textit{sk}$ and a public key $\textit{pk}$.

    \item \textit{$\mathsf{DS.Sign} \left ( \textit{sk}, m \right )$  $\longrightarrow \left ( \sigma \right ):$} A signature, denoted by the symbol $\sigma$, is produced upon the input of a secret key $sk$ and a message $m\in M$.

    \item \textit{$\mathsf{DS.Verify} \left ( \textit{pk}, \sigma, m \right )$  $\longrightarrow \left ( b \right ):$} A public key $pk$,a signature $\sigma$ and a message $m\in M$ is input, and a decision bit $b\in \left \{ 0,1 \right \}$ is produced.

\end{itemize}

\textbf{Existential Unforgeability Under a Chosen Message Attack (\textbf{EU-CMA}):} Knowing the public key $pk$, a probability polynomial-time ( probability polynomial-time, PPT) attacker is able to compute a valid signature for the new data $M'$ with negligible probability after obtaining the valid digital signature it wishes to obtain. If a digital signature scheme satisfies the above security requirements, then a valid digital signature can convince the data receiver that the data it receives has not been tampered with and the sender of the data is the owner of the corresponding public key $pk$. Next, we give the security model of \textbf{EU-CMA}. 

\begin{definition}{(\textbf{EU-CMA} Security):}
    On the subsequent experiment, the \textbf{EU-CMA} security definition of a digital signature $\mathsf{DS}$ is founded.

\fbox{%
  \parbox{0.4\textwidth}{

\vspace{1ex}\noindent\emph{$\textbf{Exp}_{\mathcal{A}, \textsf{DS}}^{\mathrm{EUF}-\mathrm{CMA}}\left(1^{\lambda}\right)$}:
\begin{enumerate}[  ]
    \item $pp \leftarrow \textsf{DS} . \textit{Setup}\left(1^{\lambda}\right);$ 
    \item $\mathcal{L}_{\textit{key}} \leftarrow \emptyset ; / / \text{KeyGen query list}$
    \item $\mathcal{L}_{\textit{sign}} \leftarrow \emptyset ; / / \text{Sign query list}$
    \item $\mathcal{L}_{\textit{corr}} \leftarrow \emptyset ; / / \text{Corrupt query list}$
    \item $\left(pk^{*}, m^{*}, \sigma^{*}\right) \leftarrow \mathcal{A}^{\mathcal{O}_{\textit{KeyGen}}(\cdot), \mathcal{O}_{\textit{Sign}}, \mathcal{O}_{\textit{Corrupt}}(\cdot)}(pp);$
    \item $\text{ if } pk^{*} \notin \mathcal{L}_{\textit{corr}} \wedge\left(pk^{*}, m^{*}\right) \notin \mathcal{L}_{\textit{sign}} \wedge$
    \item $\quad\textsf{DS}.\textit{Verify}\left(pk^{*}, \sigma^{*}, m^{*}\right)=1,$
    \item $\quad\text{return } 1.$
    \item $\text{else return } 0 .$
\end{enumerate}
where 
\begin{enumerate}[]
     \item \vspace{1ex}\noindent\emph{$\textbf{Oracle }\mathcal{O}_{\textit{KeyGen}}(i)$}:
    \item \ \ \ $(sk, pk) \leftarrow \mathcal{D} \mathcal{S} . \textit{ KeyGen }(pp);$
    \item \ \ \ $\mathcal{L}_{\textit{key}} \leftarrow \mathcal{L}_{\textit{key}} \cup\{(i, sk, pk)\};$
    \item \ \ \ $\text{return } pk.$
\end{enumerate}

\begin{enumerate}[]
     \item \vspace{1ex}\noindent\emph{$\textbf{Oracle }\mathcal{O}_{\textit{Sign}}(pk, m)$}:
     \item \ \ \ $\sigma \leftarrow \mathcal{DS} . \operatorname{Sign}(sk, m);$
     \item \ \ \ $\mathcal{L}_{\textit{sign}} \leftarrow \mathcal{L}_{\textit{sign}} \cup\{(pk, m)\};$
     \item \ \ \ $\text{return } \sigma.$
\end{enumerate}

\begin{enumerate}[]
     \item \vspace{1ex}\noindent\emph{$\textbf{Oracle }\mathcal{O}_{\textit{Corrupt}}(pk)$}:
     \item \ \ \ $\mathcal{L}_{\textit{corr}} \leftarrow \mathcal{L}_{\textit{corr}} \cup\{(pk)\};$
     \item \ \ \ $\text {return } sk.$
\end{enumerate}
  }
}

When the following advantage is negligible for any probabilistic polynomial-time adversary $\mathcal{A}$, we claim that a digital signature scheme $\mathsf{DS}$ is \textbf{EU-CMA} secure:

\begin{align*} 
\operatorname{Adv}_{\mathcal{A}, \textsf{DS}}^{\textbf{EU-CMA}}\left(1^{\lambda}\right)=\left|\operatorname{Pr}\left[\operatorname{Exp}_{\mathcal{A}, \textsf{DS}}^{\textbf{EU-CMA}}\left(1^{\lambda}\right)=1\right]\right|.
\end{align*}
\end{definition}

\subsection{Chameleon Hashes}
\textbf{Chameleon Hashes:} Compared with the traditional hash function's difficulty in finding collisions, chameleon hash can artificially set a trapdoor, and mastering this trapdoor makes it easy to find hash collisions. To a certain extent, chameleon hash destroys the two collision resistance (weak collision resistance and strong collision resistance) characteristics of the hash function, and at the same time, it also destroys the tamper-evident property of the blockchain based on the hash function. But chameleon hash also expands the application scenarios of blockchain, and it remains infeasible for ordinary users who do not know the threshold to find collisions. In other words, the security of chameleon hash can also be guaranteed. For managers holding trapdoors, if they tamper with the blocks at will, it is also possible to verify whether the hashes of two blocks are equal. Five algorithms make up a chameleon hashes $\mathsf{CH}$ \cite{krawczyk1998chameleon, ateniese2004identity} with message space $\mathbb{M}$: $ \{ \mathsf{Setup}, \mathsf{KeyGen}, \mathsf{Hash}, \mathsf{Verify}, \mathsf{Adapt}  \} $.

\begin{itemize}
    \item $\mathsf{CH.Setup} \left ( 1^{\lambda }  \right )$  $\longrightarrow \left ( \textit{pp} \right ):$ Security parameter $\lambda \in \mathbb{N}$ is taken as an input, and a public parameter $\textit{pp}$ is output. Other algorithms take the public parameter $\textit{pp}$ as an implicit input.

    \item $\mathsf{CH.KeyGen} \left ( \textit{pp} \right )$  $\longrightarrow \left ( \textit{sk}, \textit{pk} \right ):$ A public parameter $\textit{pp}$ is input, and a secret key $sk$ and a public key $\textit{pk}$ are produced.

    \item $\mathsf{CH.Hash} \left ( \textit{pk}, m, r \right )$  $\longrightarrow \left ( h \right ):$ A public key $\textit{pk}$, a message $m\in\mathcal{M}$ and a randomness $r$ is input, and a hash value $h$ is produced.

    \item $\mathsf{CH.Verify} \left ( \textit{pk}, m, h, r \right )$  $\longrightarrow \left ( b \right ):$ Following the entry of a public key $\textit{pk}$, a message $m\in \mathcal{M}$, a hash value $h$ and a randomness $r$, then provide a decision bit $b\in \left \{ 0,1 \right \}$.

    \item $\mathsf{CH.Adapt} \left ( \textit{sk}, m, h, r, m' \right )$  $\longrightarrow \left ( r' \right ):$ Following the entry of a secret key $\textit{sk}$, a message $m\in \mathcal{M}$, a hash value $h$, a randomness $r$ and a message $m'\in \mathcal{M}$, then provide a randomness $r'$.
\end{itemize}

\subsection{Timed-Release Encryption}
\textbf{Timed-Release Encryption:} Timed-release encryption is a cryptographic primitive with a specific future decryption time specified by the sender. Its time-dependent features are important in many time-sensitive real-world applications (e.g., electronic bidding \cite{dent2007revisiting}, installment payments \cite{mao2001timed}, electronic confidential files \cite{mont2003hp}). The sender sends an encrypted message to the receiver, and no user, including the receiver, can decrypt the message until the specified time. Four algorithms make up a timed-release encryption $\mathsf{TRE}$ \cite{di1999conditional, cathalo2005efficient} with message space $\mathcal{M}$ and time space $\mathcal{T} = \left [ 0, T-1 \right ]$: $\left \{ \mathsf{Setup}, \mathsf{Ext}, \mathsf{Enc}, \mathsf{Dec} \right \} $.

\begin{itemize}
    \item $\mathsf{TRE.Setup} \left ( 1^{\lambda }, T  \right )$  $\longrightarrow \left ( \textit{pp} \right ):$ Security parameter $\lambda \in \mathbb{N}$ and a  time instant $T$ are taken as the input, and a public parameter $\textit{pp}$ is produced.
 
    \item $\mathsf{TRE.Ext} \left ( \textit{mpk}, \textit{msk}, t \right )$  $\longrightarrow \left ( k_t \right ):$ A master public key $\textit{mpk}$, a master secret key $\textit{msk}$ and $t \in \mathcal{T}$ are input, and a time instant key (TIK) $k_t$ is produced.

    \item $\mathsf{TRE.Enc} \left ( \textit{mpk}, m, \left [ t_{0}, t_{1} \right ]  \right )$  $\longrightarrow \left ( C \right ):$ A master public key $\textit{mpk}$, a message $m\in \mathcal{M}$ and a Decryption Time Interval (DTI) $\left [ t_{0}, t_{1} \right ] \subseteq \mathcal{T}$ are input, and a ciphertext $C$ is produced.

    \item $\mathsf{TRE.Dec} \left ( \textit{mpk}, C, k_t \right )$  $\longrightarrow \left ( m/\bot \right ):$ Following the entry of a master public key $\textit{mpk}$, a ciphertext $C$ and a key $k_t$, then provide either a message $m$ or a failure symbol $\bot$.
\end{itemize}

\section{Problem Formulation}\label{problem-formulation}
In this section, we first define the system model and then give detailed descriptions of the threat model and security model.

\subsection{System Model}

As shown in \Cref{fig:System Model}, \textsf{Epass} is composed of a certificate authority (CA), users, providers, servers, and miners.

\begin{figure}[htpb!]
    \centering
    \includegraphics[width=0.48\textwidth]{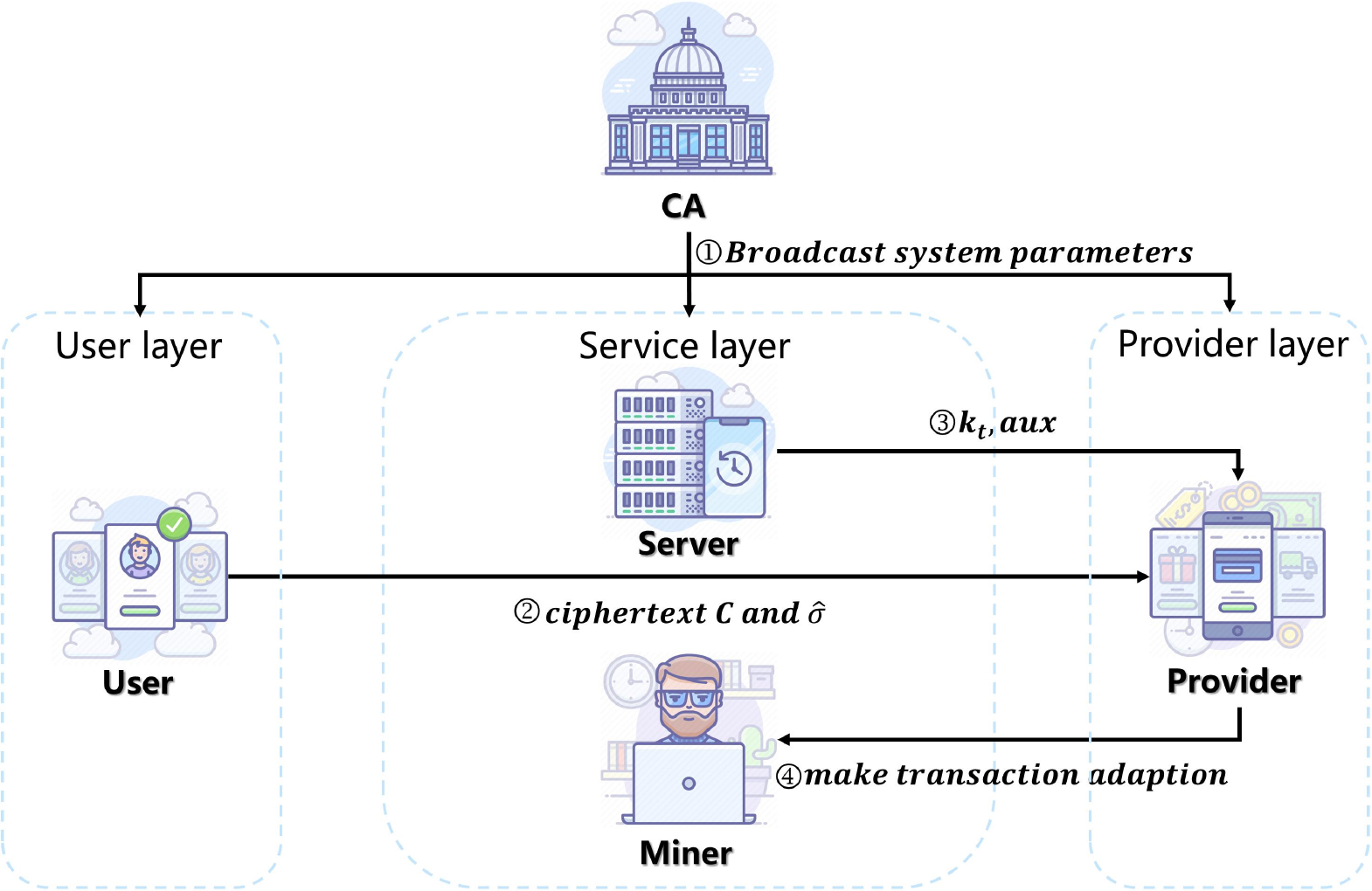}
    \caption{System Model}\label{fig:System Model}
\end{figure}

\begin{itemize}
    \item \textsf{Epass} needs to be initialized, and CA is the blockchain administrator who must broadcast the system parameters to the other participants.
    \item Users are participants in the chain and they are allowed to choose between two types of payment forms. One is instant payment, and transactions generated in this way cannot be rewritten. The other is asynchronous payment, which allows the specified provider to rewrite the content of the transaction.
    \item Providers are participants in the chain who provide asynchronous payment services to users and collect their fees when a specified time is reached.
    \item Servers are participants in the chain that broadcasts specific time nodes and provides additional auxiliary information to the provider.
    \item Miners are participants in the chain, independent and interconnected nodes that validate transactions and add them to the existing distributed public ledger.
\end{itemize}

\subsection{Threat Model}
In our proposed scheme, CA is considered to be fully trusted and miners are considered to be majority trusted, as in normal blockchain systems. The adversaries are classified into three categories based on their capabilities, i.e., intended-but-curious providers, honest-but-curious cloud and external adversary. The threat model is described in detail as follows.

\begin{itemize}
	
	\item \textbf{Intended-but-Curious Providers}: The provider performs the role of a receiver and a provider in our system. As a semi-curious participant, the provider expects to decrypt messages outside of the expected time interval and tries to modify the wrong transaction amount.
	
	\item \textbf{Honest-but-Curious Cloud}: The cloud server holds the private key $sk_s$ and is responsible for providing TIK $k_t$ and auxiliary information $aux$ at fixed intervals throughout the process. The cloud is semi-trustworthy and follows our protocol, but will attempt to launch attacks to compromise confidential messages, such as ciphertext-only attacks.
	
	\item \textbf{External Adversary}: An external attacker is neither the intended receiver nor aware of the $sk_s$, but he/she can eavesdrop during system communication and obtain the ciphertext by launching a ciphertext-only attack.
	
\end{itemize}

\subsection{Security Model}
We define a model for \textbf{IND-CPA} and \textbf{EU-CMA} security of \textsf{Epass} scheme.

\begin{definition}If the advantage of all adversaries in the game is negligible at most, then \textsf{Epass} is \textbf{IND-CPA} secure.

    \textbf{Setup.} The challenger $\mathcal{C}$ runs $\textsf{Epass}.Setup ( 1^{\lambda } )$ to generate a master secret key $msk$ and a master public key $\textit{mpk}$. Then, the master secret key $msk$ is given to the adversary $\mathcal{A}$.

    \textbf{Phase 1.} At any time $t\in T$, $\mathcal{A}$ can adaptively issue a TIK extraction query to oracle. Oracle will respond to each query with $k_t$.

    \textbf{Challenge.} $\mathcal{A}$ chooses two messages $m_0$ and $m_1 \in \mathcal{M}$ and a time interval $[t_{0},  t_{1}] \subseteq T$ with the constraint that for all queries t in Phase 1, $t\notin [t_{0},  t_{1}]$. $\mathcal{A}$ passes $m_0$, $m_1$, $[t_{0},  t_{1}]$ to $\mathcal{C}$. $\mathcal{C}$ chooses a random bit $b$ and computes 
    \begin{align*}
    c' = \textsf{Epass}.TrCreat \left ( \textit{mpk}, \textit{sk}_u, \left ( \textit{ID}, \textit{tx}_{\textit{ID}}\right ) , \textit{sk}_{h}, T \right ).\end{align*}
    $c'$ is passed to $\mathcal{A}$.

    \textbf{Phase 2.} $\mathcal{A}$ continues to query the TIK extract oracle using the same restrictions as the Challenge phase.

    \textbf{Guess.} The adversary outputs its guess $b'$ for $b$.

    The output of this game is defined as 1 when $b'$ = $b$ and 0 otherwise. If the output of the game is 1, we say that $\mathcal{A}$ succeeds. We denote the advantage of $\mathcal{A}$ winning the game by:
    \begin{align*} Adv_{\mathcal{A}}(\kappa)=\left|\operatorname{Pr}\left[b^{\prime}=b\right]-\frac{1}{2}\right|.
    \end{align*}

    We solve the security problem of \textbf{IND-CCA} by extending the definition. First, consider a Decrypt oracle in addition. On input a pair $(c, t)$, it gets the response $kt$ after passing $t$ to the TIK extraction oracle, where $c$ represents the ciphertext and $t\in T$. Then, the Decrypt oracle returns a message $m$ or a failure symbol $\perp$ to the adversary by running $\textsf{Epass}.ReleasedDec \left ( \hat{\sigma}, \textit{pk}_u, \textit{pk}_{S}, m, \textit{aux}, C, k_t \right )$. The Decrypt oracle can adaptively issue queries $(c, t)$ in both Phase 1 and Phase 2, but will be restricted in the latter phase, i.e., the adversary cannot make decrypt queries $(c, t)$, where $c = c'$, $t \in T$, $c'$ represents the challenge ciphertext and $T$ represents the time interval. Under this restriction, the adversary cannot win the game in a simple way.

\end{definition}

\begin{definition}(\textbf{EU-CMA} Security): If the advantage $\operatorname{Adv}_{\mathcal{A}, \mathcal{\textsf{Epass}}}^{\mathrm{EU-CMA}}(1^{\lambda})=   \operatorname{Adv}_{\mathcal{A}, \\textsf{DS}}^{\textbf{EU-CMA}}(1^{\lambda}) $ is negligible for any probabilistic polynomial-time adversary $\mathcal{A}$, then \textsf{Epass} is \textbf{EU-CMA} secure. $\mathcal{O}$ is defined as the set of oracles, including:  a provider key generation oracle $\mathcal{O}_{\textit{KeyGen}_{p}}(\cdot)$,  a provider corrupt oracle $\mathcal{O}_{\textit{Corrupt}_{p}}(\cdot)$, a user key generation oracle $\mathcal{O}_{\textit{KeyGen}_{u}}(\cdot)$, a user corrupt oracle $\mathcal{O}_{\textit{Corrupt}_{u}}(\cdot)$, a server key generation oracle $\mathcal{O}_{\textit{KeyGen}_{S}}(\cdot)$, a server corrupt oracle $\mathcal{O}_{\textit{Corrupt}_{S}}(\cdot)$, a hash oracle $\mathcal{O}_{\textit{TrCreat}}(\cdot, \cdot)$, an aggregate oracle $\mathcal{O}_{\textit{Aggregate}}( \cdot,  \cdot )$ and an adaption oracle $\mathcal{O}_{\textit{Adapt}}( \cdot, \cdot,  \cdot, \cdot,  \cdot)$.
\end{definition}

\subsection{Design Goals}
Based on the requirements of the above models, our design goals are divided into two aspects: privacy and efficiency.
\begin{itemize}
    \item \textit{Privacy}: The privacy of user transactions should be guaranteed in asynchronous payments. The deferred payment transactions generated by the user should not be made public to prevent malicious third parties from analyzing the user's financial status accordingly.

    \item \textit{Efficiency}: Reasonable performance overhead and scalability should be ensured. \textsf{Epass} saves time and computational power compared to existing works while satisfying basic transaction processing requirements.
\end{itemize}

\section{Privacy-Preserving Blockchain-based Asynchronous Payment Scheme} In this section, we give the formal definition and scheme description of \textsf{Epass}.

\label{construction}
\subsection{Formal Definition}
\textsf{Epass} is made up of the ten algorithms listed below:

$\textsf{Epass}.\textit{Setup}\left ( 1^{\lambda }  \right ) \rightarrow \left ( \textit{pp}, \textit{msk}, \textit{mpk} \right ):$ CA manages the setup algorithm. It accepts a security parameter $\lambda \in \mathbb{N}$ as input. It produces a public parameter $\textit{pp}$, a master secret key $\textit{msk}$, and a master public key $\textit{mpk}$, with $\textit{pp}$ and $\textit{mpk}$ serving as implicit inputs to all other algorithms.
	
$\textsf{Epass}.\textit{UserKeyGen} \left ( \textit{pp}  \right ) \rightarrow \left ( \textit{sk}_{u}, \textit{pk}_{u}\right ):$ Each user executes the user setup algorithm. It accepts a public parameter $\textit{pp}$ as input and returns a secret key $\textit{sk}_u$ and a public key $\textit{pk}_u$ in the form of outputs.

$\textsf{Epass}.\textit{ProviderKeyGen} \left ( \textit{pp}  \right ) \rightarrow \left ( \textit{sk}_{p}, \textit{pk}_{p} \right ):$ Each provider executes the provider setup algorithm. A secret key $\textit{sk}_p$ and a public key $\textit{pk}_p$ are produced from an input of a public parameter $\textit{pp}$.

$\textsf{Epass}.\textit{ServerKeyGen} \left ( \textit{pp}  \right ) \rightarrow \left ( \textit{sk}_{S}, \textit{pk}_{S} \right ):$ The server manages the server setup algorithm. A secret key $\textit{sk}_{S}$ and a public key $\textit{pk}_{S}$ are produced from an input of a public parameter $\textit{pp}$.

$\textsf{Epass}.\textit{TrCreat} \left ( \textit{mpk}, \textit{sk}_u, \left ( \textit{ID}, \textit{tx}_{\textit{ID}}\right ) , \textit{sk}_{h}, T \right )$ $\rightarrow \left ( h, r, \sigma, C \right ):$ Each user executes the hash algorithm. The inputs are a message with a transaction identity $\textit{ID}$ and its content $\textit{tx}_{\textit{ID}}$, a master public key $\textit{mpk}$, a pair of secret keys $\textit{sk}_u$ and $\textit{sk}_h$, and a decryption time $T$. It generates a ciphertext $C$, a signature $\sigma$, a hash value $h$, and randomness $r$.

$\textsf{Epass}.\textit{Aggregate} \left ( \textit{pk}_{u}, \left \{ \left ( \left ( \textit{ID}_{i} , r_{i}  \right ), \sigma_{i} \right )     \right \} _{i}    \right ) \rightarrow \left ( \hat{\sigma } /\bot  \right ):$ All input signatures $\sigma_{i}$ are first verified by the signature aggregation process, which outputs $\bot$ if any of these verifications are unsuccessful. If not, it produces the aggregated signature $\hat{\sigma }$.

$\textsf{Epass}.\textit{Ext} \left ( \textit{pk}_{S}, \textit{sk}_{S}, \textit{pk}_{u}, \left \{ \textit{tx}_{i}  \right \} _{i\in \left [ \ell  \right ] }, j\in \left [ \ell  \right ] \right ) \rightarrow \left ( k_t, \textit{aux} \right ):$ On the server, the extraction algorithm is executed. It accepts a collection of transactions $\left \{ \textit{tx}_{i}  \right \} _{i\in \left [ \ell  \right ] }$, public keys $\textit{pk}_{S}$ and $\textit{pk}_{u}$, a secret key $\textit{sk}_{S}$, and outputs a TIK $k_t$ and the auxiliary information $\textit{aux}$.

$\textsf{Epass}.\textit{ReleasedDec} \left ( \hat{\sigma}, \textit{pk}_u, \textit{pk}_{S}, m, \textit{aux}, C, k_t  \right ) \rightarrow \left ( \textit{sk}_{h} /\bot\right ):$ Each provider executes the timed-release decryption algorithm. It accepts the following as inputs: public keys $\textit{pk}_u$ and $\textit{pk}_{S}$, an aggregate signature $\hat{\sigma}$, a message with a transaction identity $\textit{ID}$ and its content $\textit{tx}_{\textit{ID}}$ and $r$, the auxiliary information $\textit{aux}$, a ciphertext $C$, and the TIK $k_t$. And if any of these verification fail, outputs $bot$. If not, the aggregated signature $\textit{sk}_{h}$ is output.

$\textsf{Epass}.\textit{Adapt} \left ( \textit{mpk}, \textit{sk}_{p}, \left ( \textit{ID}, \textit{tx}_{\textit{ID}} \right ),  h, r, \left ( \textit{ID}, \textit{tx} _{\textit{ID}}' \right )  \right )$ $\rightarrow \left ( r', \sigma' \right ):$ Each provider executes the adapt algorithm. It accepts the following inputs: the master public key $\textit{mpk}$, the secret key $\textit{sk}_{p}$, the message's transaction identity $\textit{ID}$ and its content $\textit{tx}_{\textit{ID}}$, the hash value $h$, the randomness $r$, and the message's transaction identity $\textit{ID}$ and its content $\textit{tx} _{\textit{ID}}'$. It generates two values: randomness $r'$ and a signature $\sigma'$.

$\textsf{Epass}.\textit{Verify} \left ( \textit{mpk}, \textit{pk}_{u}, \left ( \textit{ID}, \textit{tx}_{\textit{ID}} \right ), h, r, \sigma_{\textit{ID}}  \right ) \rightarrow \left ( b \right ):$ Miners operate the verification algorithm. A master public key $\textit{mpk}$, a public key $pk_{u}$, a message including a transaction identity $\textit{ID}$ and its content $\textit{tx}_{\textit{ID}}$, a hash value $h$, a randomness $r$, and a signature $\sigma_{\textit{ID}} $ are all inputs required. It generates a judgment bit $b\in \left \{ 0,1 \right \}$ indicating the validity of the transaction $\left ( \textit{ID}, \textit{tx}_{\textit{ID}} \right )$.

\subsection{Proposed Scheme}
The four phases of \textsf{Epass} are system initialization, transaction making, transaction rewriting, and transaction verification.

\begin{itemize}
    \item[1)] \textbf{System Initialization:} The initialization of \textsf{Epass} is displayed in \Cref{fig:System Initialization}. This process can be more specifically divided into system setup, user key generation, provider key generation, and server key generation.
    
    \begin{figure}[htpb!]
        \centering
        \includegraphics[width=0.48\textwidth]{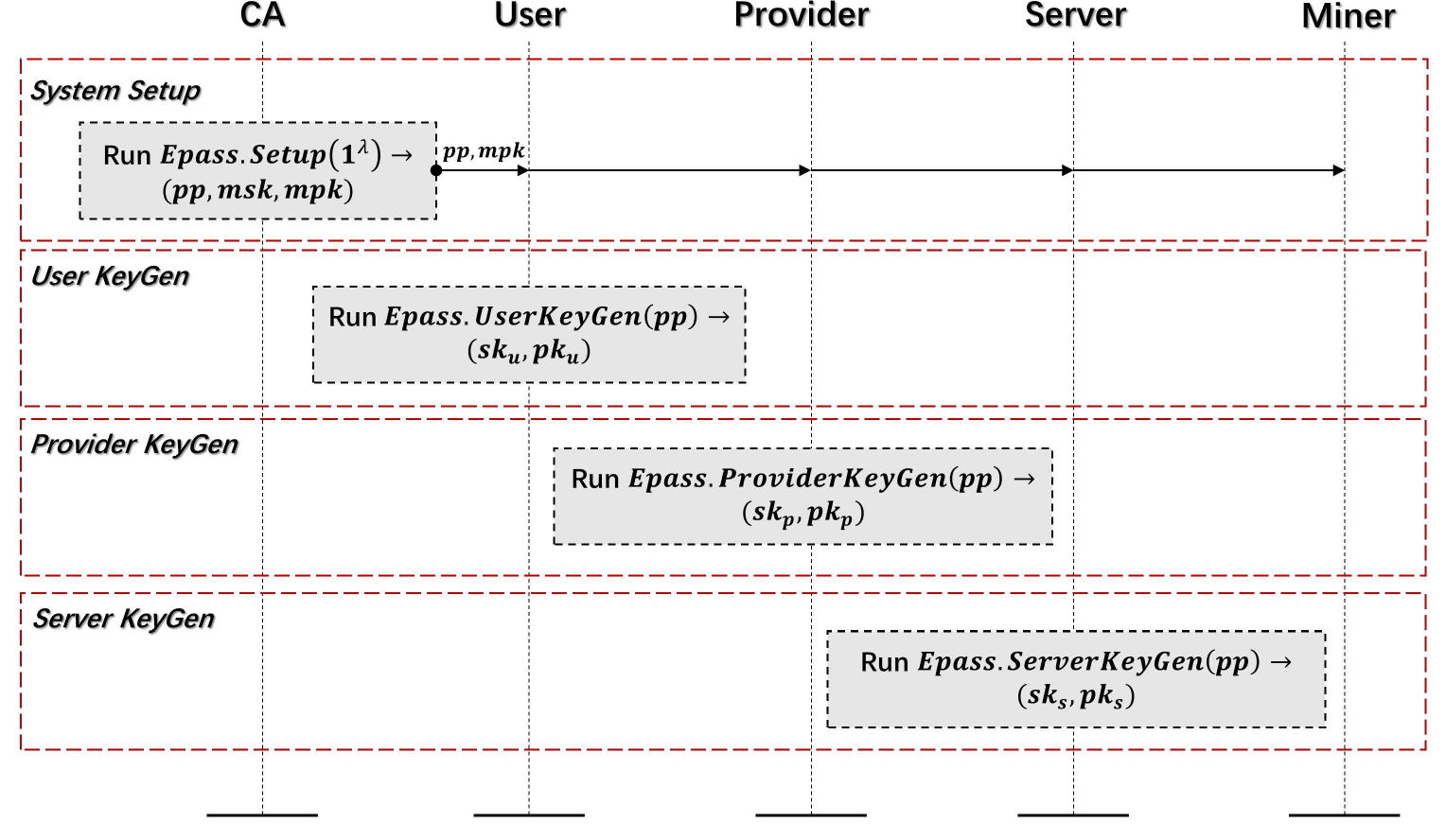}
        \caption{System Initialization}\label{fig:System Initialization}
    \end{figure}
    
    $\textsf{Epass}.\textit{Setup}\left ( 1^{\lambda }  \right ) \rightarrow \left ( \textit{pp}, \textit{msk}, \textit{mpk} \right ):$ Given a security parameter $\lambda \in \mathbb{N}$ and the upper bound on number of aggregations $B$, set $\textit{pp}_\mathsf{DS}=(p,\mathbb{G},\mathbb{G}_T,g,\hat{e})$ as the bilinear group used in our construction, where $\hat{e}:\mathbb{G}\times\mathbb{G}\rightarrow\mathbb{G}_T$, and $\mathbb{G}$ is a prime $p$ order group. The public parameters of a chameleon hash is $\textit{pp}_\mathsf{CH} \gets \mathsf{CH}.\textit{Setup} \left ( 1^{\lambda }  \right )$. Selects a random number $\alpha \gets \mathbb{Z}_{p}^{*}$, and samples the public parameters for message hashing as $\textit{hk}\gets  HGen\left ( 1^{\lambda }  \right )$.  It sets the key pair as $\textit{pk}_{\alpha } =( \textit{pp}_\mathsf{DS}, \{ g^{\alpha ^{i} }  \}_{i\in \left [ B \right ] }   ) $, and $\textit{sk}_{\alpha } =\left ( \textit{pp}_\mathsf{DS}, \alpha  \right )$. Then, a collision-resistant hash function $H:\left \{ 0, 1 \right \} ^{*} \to \mathbb{Z}_{p}$ is chosen, and the following cryptographic hash functions $H_{1}:\left \{ 0, 1 \right \}^{*} \to G^{*}$, $H_{2}:G_{T}^{*}\to \left \{ 0, 1 \right \}^{n}$ are constructed. The algorithm outputs a public parameter $\textit{pp} = (\textit{pp}_\mathsf{DS}, \textit{pp}_\mathsf{CH}, \textit{hk})$, a master secret key $\textit{msk} = \textit{sk}_{\alpha }$, and a master public key $\textit{mpk} = (\textit{pk}_{\alpha }, H)$, where $\textit{pp}$ and $\textit{mpk}$ are made available to the public.


   $\textsf{Epass}.\textit{UserKeyGen} \left ( \textit{pp}  \right ) \rightarrow \left ( \textit{sk}_{u}, \textit{pk}_{u}\right ):$ The user key generation algorithm initializes a signature key-pair $\textit{sk}_{\beta} =\left ( \textit{pp}_\mathsf{DS}, \textit{hk}, \beta \gets \mathbb{Z}_{p}^{*}  \right )$, and $\textit{pk}_{\beta} =( \textit{pp}_\mathsf{DS}, \textit{hk},  \{ g^{\beta ^{i} }  \}_{i\in \left [ B \right ] } ) $, where $B$ represents the upper bound for deferred payment transactions. It then initializes a chameleon hash key-pair $\left ( \textit{sk}_{h}, \textit{pk}_{h} \right ) \gets \mathsf{CH.KeyGen} \left ( \textit{pp}_{\mathsf{CH}}  \right )$, a timed-release key pair $\textit{sk}_{\textit{tre}}=s\gets \mathbb{Z}_{p}^{*}$, $\textit{pk}_{\textit{tre}}=(\textit{sg}, s\alpha'g)$ and generate a local verification key $\textit{pk}_{\textit{local}} =\left(\Pi, \mathrm{hk}, g^{\alpha}\right)$. The algorithm returns a secret key $\textit{sk}_u = (\textit{sk}_{\beta}, \textit{sk}_{\textit{tre}})$ and a public key $\textit{pk}_u = \left (\textit{pk}_{\beta}, \textit{pk}_{h}, \textit{pk}_{\textit{tre}}, \textit{pk}_{\textit{local}} \right )$.


    $\textsf{Epass}.\textit{ProviderKeyGen} \left ( \textit{pp}  \right ) \rightarrow \left ( \textit{sk}_{p}, \textit{pk}_{p} \right ):$ The provider key generation algorithm initializes a signature key-pair $\textit{sk}_{p'} = ( \textit{pp}_\mathsf{DS}, \textit{hk}, \mu \gets \mathbb{Z}_{p}^{*} )$, and $\textit{pk}_{p'} =( \textit{pp}_\mathsf{DS}, \textit{hk}, \{ g^{{\mu} ^{i} } \}_{i\in  [B]}) $. The algorithm outputs two keys: a secret one $\textit{sk}_p = \textit{sk}_{p'}$ and a public one $\textit{pk}_p = \textit{pk}_{p'}$.


    $\textsf{Epass}.\textit{ServerKeyGen} \left ( \textit{pp}  \right ) \rightarrow \left ( \textit{sk}_{S}, \textit{pk}_{S} \right ):$ The server key generation algorithm initializes a key-pair as $\textit{sk}_{\alpha'} =\left ( \textit{pp}_\mathsf{DS}, \alpha'\gets \mathbb{Z}_{p}^{*}  \right )$, and $\textit{pk}_{\alpha' } =( \textit{pp}_\mathsf{DS}, g\gets \mathbb{G},  Z=\alpha'g )$, $g$ and $\alpha'g$ are made public. The algorithm returns a secret key $\textit{sk}_{S} = \textit{sk}_{\alpha'}$ and a public key $\textit{pk}_{S} = Z$.


    \item[2)] \textbf{Transaction Making:} \Cref{fig:Transaction Making} shows the \textsf{Epass} transaction making. Each user creates two kinds of transactions during this phase: redactable and immutable. The redactable transactions can be modified by the provider at a specific time. After that, the server generates TIK $k_t$ and auxiliary information $\textit{aux}$.
    
    \begin{figure}[htpb!]
        \centering
        \includegraphics[width=0.48\textwidth]{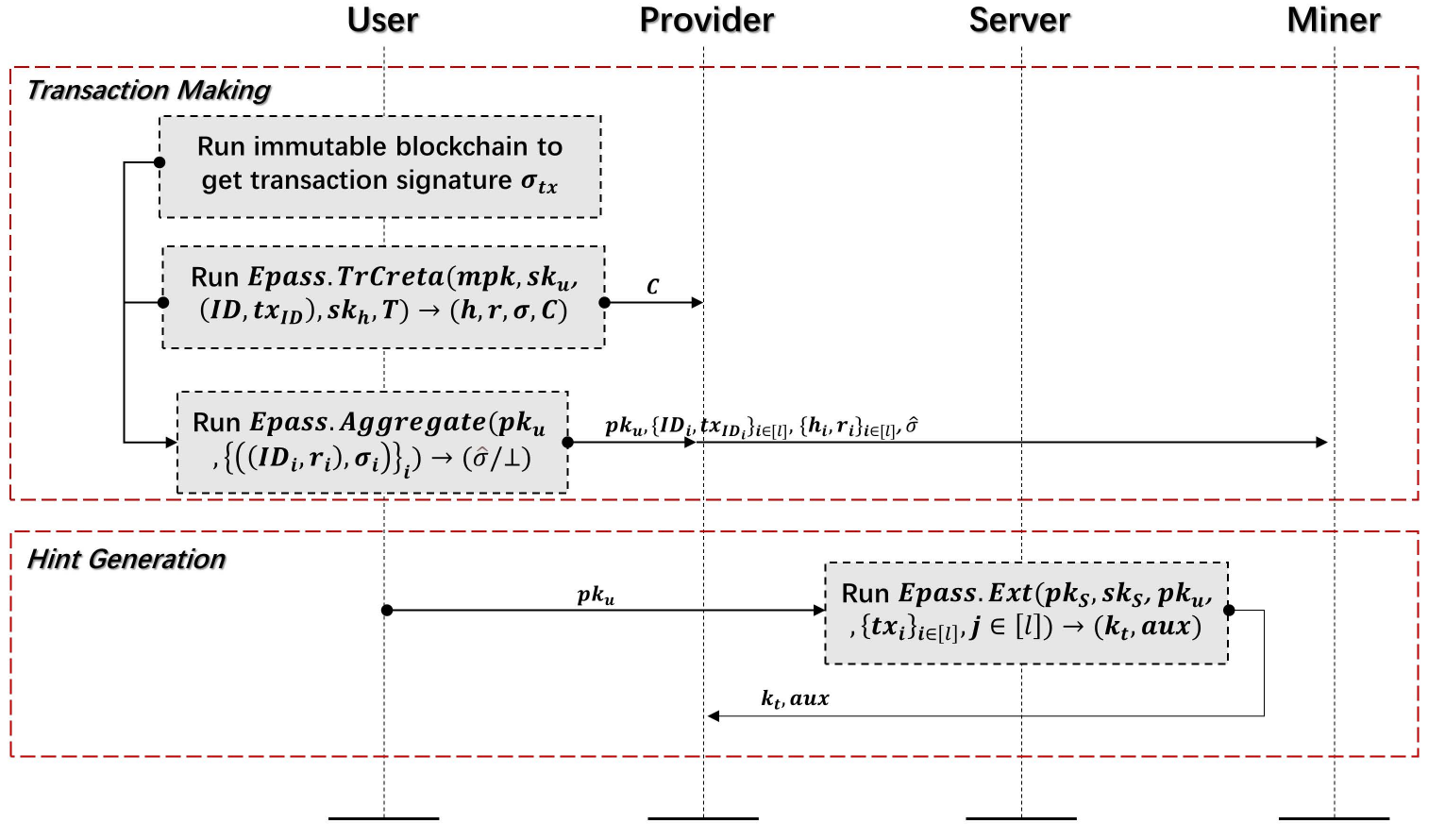}
        \caption{Transaction Making}\label{fig:Transaction Making}
    \end{figure}

    \begin{itemize}
        \item \textbf{Ordinary and deferred transaction making:} The user can generate ordinary transactions as in a traditional immutable blockchain, or run $\textsf{Epass}.\textit{TrCreat}  ( \textit{mpk}, \textit{sk}_u,  ( \textit{ID}, \textit{tx}_{\textit{ID}} ) , \textit{sk}_{h}, T  )$ to generate deferred transactions, as shown in Algorithm 1. In a normal transaction, the user generates a transaction $\textit{tx}$ and the corresponding signature $\sigma_{\textit{tx}}$ by means of the key $\textit{sk}_u$ and the traditional hash function. For deferred transactions, the user generates the transaction $(\textit{ID},\textit{tx}_{\textit{ID}})$ and the signature $(h,r, \sigma_{\textit{ID}})$ by means of the key $sk_u$ and the chameleon hash function. Finally, the aggregated signature $\hat{\sigma}$ is generated by running $\textsf{Epass}.\textit{Aggregate}  ( \textit{pk}_{u},  \{  (  ( \textit{ID}_{i} , r_{i}   ), \sigma_{i}  )      \} _{i}     ) $, as shown in Algorithm 2. The user propagates to the blockchain ecosystem the public key $\textit{pk}_p$, the transaction set $ \{ \textit{tx}_{i}   \} _{i\in [ \ell ] }$, and the aggregated signature $\hat{\sigma }$.

\begin{algorithm}\label{algr5}
	\caption{$\textsf{Epass}.\textit{TrCreat} ( \textit{mpk}, \textit{sk}_u, ( \textit{ID}, \textit{tx}_{\textit{ID}}) , \textit{sk}_{h}, T  )$}
	\KwIn{Master public key $\textit{mpk}$, secret keys $\textit{sk}_u$ and $\textit{sk}_h$, a decryption time $T$, and a message containing the transaction identity $\textit{ID}$ and the transaction's content $\textit{tx}_{\textit{ID}}$.}
 
    \KwOut{The randomness $r$, a hash value $h$, a randomness $r$, a signature $\sigma$ and ciphertext $C$.}
    Choose a random number $r$;\\
    Compute a hash value: 
    \begin{align*} 
    h\gets \mathsf{CH}.\textit{Hash}\left ( \textit{pk}_{h}, (\textit{ID}, \textit{tx}_{\textit{ID}}), r \right );
    \end{align*}\\
    Generate a signature: $\sigma _{\textit{ID}} = g^{(1/\textit{sk}_{u}+h)}$;\\
    Check if $\hat{e}(\textit{sg}, \alpha 'g)\stackrel{?}{=}\hat{e}(g, s\alpha 'g)$;\\
    If so, random choose $r_0$, and compute $r_0g$ and $r_0s\alpha'g$;\\
    Compute 
    \begin{align*}
    K=\hat{e}\left(r_{0} s\alpha 'g, H_{1}(T)\right)=\hat{e}\left(g, H_{1}(T)\right)^{r_{0} s\alpha '};
    \end{align*}\\
    Compute the ciphertext:
    \begin{align*}
    C=<r_0g, sk_h \oplus H_{2}(K)>.
    \end{align*}\\
    \textbf{return} The hash value $h$, the random number $r$, the signature $\sigma _{\textit{ID}}$ and the ciphertext $C$.
\end{algorithm}

\begin{algorithm}\label{algr6}
	\caption{$\textsf{Epass}.\textit{Aggregate}  ( \textit{pk}_{u},  \{  ( ( \textit{ID}_{i} , r_{i}  ), \sigma_{i} ) \} _{i} )$}
	\KwIn{Public key $\textit{pk}_{u}$, signatures $\sigma_{i}$.}
    \KwOut{The aggregated signature $\hat{\sigma }$.}
    \eIf{ $e(\sigma, g^{\alpha} g^{h}) = e(g, g).$}
    {Compute:
    $\gamma _{i}=\frac{1}{\prod_{i \neq j}\left(x_{i}-x_{j}\right)},$
    \\
    and compute the aggregated signature:
    \begin{align*}
    \hat{\sigma}=\prod_{i} \sigma_{i}^{\gamma i}.
    \end{align*}}
    {\textbf{return} $\bot.$}
\end{algorithm}
        
        \item \textbf{Extraction:} At a time instant $T\in \left \{ 0, 1 \right \} ^{*}$, the server publishes $k_t=\alpha 'H_{1} (T)$, every user can verify its authenticity by checking $\hat{e}\left(\alpha 'g, H_{1}(T)\right)=\hat{e}\left(g, \alpha 'H_{1}(T)\right).$ Then, the server generates TIK $k_t$ and auxiliary information $aux$ by running $\textsf{Epass}.\textit{Ext}  ( \textit{pk}_{S}, \textit{sk}_{S}, \textit{pk}_{u},  \{ \textit{tx}_{i}  \} _{i\in [\ell] }, j\in[ \ell]) $, as shown in Algorithm 3. These information will be published at a specific time for the provider to verify the signature and perform decryption operations.

    \begin{algorithm}\label{algr7}
    	\caption{$\textsf{Epass}.\textit{Ext} ( \textit{pk}_{S}, \textit{sk}_{S}, \textit{pk}_{u}, \{ \textit{tx}_{i} \} _{i\in  [ \ell ] })$}
    	\KwIn{Public keys $\textit{pk}_{S}$ and $\textit{pk}_{u}$, a secret key $\textit{sk}_{S}$, a set of transactions $ \{ \textit{tx}_{i}   \} _{i\in  [ \ell  ] }$.}
        \KwOut{The TIK $k_t$ and the auxiliary information $\textit{aux}$.}
        \eIf{ $\hat{e}(\alpha 'g, H_{1}(T))=\hat{e}(g, \alpha 'H_{1}(T)).$}
        {Generate the auxiliary informations: 
        \begin{align*} 
        \textit{aux}_{j, 1}=g^{\prod_{i \neq j}(\alpha+h_{i})},
        \end{align*}
        \begin{align*}
        \textit{aux}_{j, 2}=g^{\alpha \prod_{i \neq j}(\alpha+h_{i})};
        \end{align*}\\
        Compute the following polynomial $P$ to obtain the coefficients $\{\tilde{\delta _{i} }  \in \mathbb{Z}_{p}\}_{i \in[\ell-1]}:$
        \begin{align*}
        P_{\left\{x_{i}\right\}_{i \in[\ell] \backslash\{j\}}}(y)&=\prod_{i \in[\ell] \backslash\{j\}}\left(y+x_{i}\right)\\
        &=\sum_{i=0}^{\ell-1} \widetilde{\delta}_{i} y^{i}(\bmod p);
        \end{align*}\\ 
        Compute 
        \begin{align*}
        \operatorname{\textit{aux}}_{j, 1}=\prod_{i=0}^{\ell-1}(g^{\alpha^{i}})^{\widetilde{\delta }_{i}},
        \end{align*}
        \begin{align*}
        \operatorname{\textit{aux}}_{j, 2}=\prod_{i=0}^{\ell-1}(g^{\alpha^{i+1}})^{\widetilde{\delta _{i}}};
        \end{align*}\\
        Outputs the auxiliary informations 
        \begin{align*}
        \operatorname{\textit{aux}}_{j}=(\operatorname{\textit{aux}}_{j, 1}, \operatorname{\textit{aux}}_{j, 2}).
        \end{align*}}
        {\textbf{return} Null}
    \end{algorithm}
    \end{itemize}

    \item[3)] \textbf{Transaction Rewriting:} \Cref{fig:Transaction Rewriting} shows the \textsf{Epass} transaction rewriting. In this phase, the provider decompresses the subset of transaction signatures that need to be paid asynchronously and performs the decryption operation based on the TIK $k_t$ and auxiliary information $aux$ provided by the server. The designated provider can rewrite the redactable transaction.

    \begin{figure}[htpb!]
        \centering
        \includegraphics[width=0.48\textwidth]{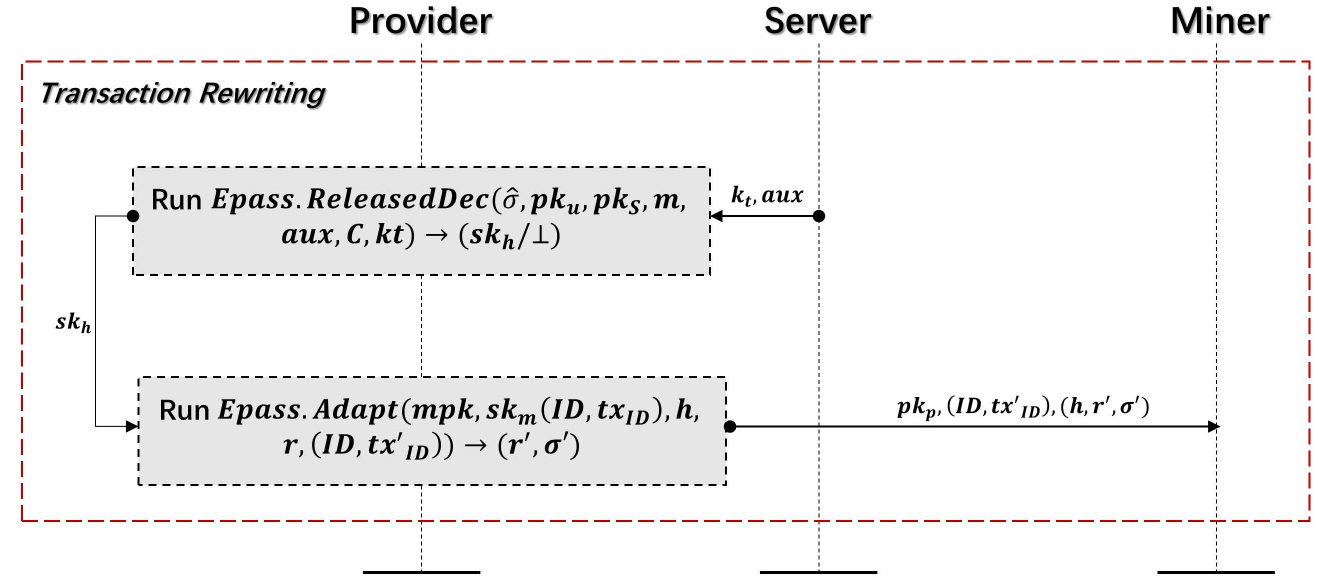}
        \caption{Transaction Rewriting}\label{fig:Transaction Rewriting}
    \end{figure}

    \begin{itemize}
        \item \textbf{Timed-release decryption:} The provider holds the TIK $k_t$ and auxiliary information $\textit{aux}$ provided by the server, and gets the secret key  $\textit{sk}_h$ and the subset of transactions that need to be rewritten by running the $\textsf{Epass}.\textit{ReleasedDec} \left ( \hat{\sigma}, \textit{pk}_u, \textit{pk}_{S}, m, \textit{aux}, C, k_t \right )$, as shown in Algorithm 4. After that, the provider can use this secret key to rewrite the transaction.

    \begin{algorithm}\label{algr8}
    	\caption{$\textsf{Epass}.\textit{ReleasedDec} ( \hat{\sigma}, \textit{pk}_u, \textit{pk}_{S}, m, \textit{aux}, C, k_t  )$}
    	\KwIn{Public keys $\textit{pk}_u$ and $\textit{pk}_{S}$, an aggregate signature $\hat{\sigma}$, a message containing the transaction's identity $\textit{ID}$ and content $\textit{tx}_{\textit{ID}}$ and $r$, the auxiliary information $\textit{aux}$, a ciphertext $C$ and the TIK $k_t$.}
        \KwOut{The aggregated signature $\textit{sk}_{h}$ or $\bot$.}
        Compute the message hash set as: 
        \begin{align*} 
        \{ h_m \}_{i} \gets \mathsf{CH}.\textit{Hash}( \textit{pk}_{h}, \{ m \}_{i} );
        \end{align*}\\
        \eIf{$e(\widehat{\sigma}, \operatorname{\textit{aux}}_{1}^{ \{ h_m \}_{i}} \operatorname{\textit{aux}}_{2})\stackrel{?}{=}e(g, g)$ and $e(g^{\alpha}, \mathrm{\textit{aux}}_{1})\stackrel{?}{=}e(g, \mathrm{\textit{aux}}_{2})$} 
        {{\textbf{return} 1 to signal that the signature is valid;}\\
        Parse $C$ as 
        \begin{align*} 
        <r_0g, \rho=\textit{sk}_h \oplus H_{2}(K)>;
        \end{align*}\\
        Compute 
        \begin{align*} 
        K^{\prime}=\hat{e}(r_0g,\alpha ' H_{1}(T))^{s}=\hat{e}(g, H_{1}(T))^{r_{0} s\alpha '}=K;
        \end{align*}\\
        Recover $\textit{sk}_h$ by computing $\rho  \oplus H_{2}(K)$.\\
        {\textbf{return} $\textit{sk}_h$}}
        {\textbf{return} $\bot$}
    \end{algorithm}

        \item \textbf{Transaction rewriting:} The provider runs the $\textsf{Epass}.\textit{Adapt} \left ( \textit{mpk}, \textit{sk}_{p}, \left ( \textit{ID}, tx_{\textit{ID}} \right ), h, r, \left ( \textit{ID}, \textit{tx} _{\textit{ID}}' \right ) \right )$ algorithm to generate random numbers $ r'$ and signatures $\sigma' $, as shown in Algorithm 5. Then, the provider broadcasts the public key $pk_p$, the transaction $(\textit{ID},\textit{tx} _{\textit{ID}}')$ and the signature $(h,r',\sigma')$.

    \begin{algorithm}\label{algr9}
    	\caption{$\textsf{Epass}.\textit{Adapt}  ( \textit{mpk}, \textit{sk}_{p},  ( \textit{ID}, \textit{tx}_{\textit{ID}} ),  h, r,$\\ $ ( \textit{ID}, \textit{tx} _{\textit{ID}}' )  )$}
    	\KwIn{The master public key $\textit{mpk}$, a secret key $\textit{sk}_{p}$, a message including a transaction identity $\textit{ID}$ and its content $\textit{tx}_{\textit{ID}}$, a hash value $h$, a randomness $r$ and a message including a transaction identity $\textit{ID}$ and its content $\textit{tx'} _{\textit{ID}}$.}
        \KwOut{A randomness $r'$ and a signature $\sigma'$.}
        Generate a randomness $r'$:
        \begin{align*} 
        r^{\prime} \leftarrow \mathsf{CH}.\textit{Adapt}\left(\textit{s k}_{h},\left(\textit{ID}, \textit{tx}_{\textit{ID}}\right), h, r,\left(\textit{ID}, \textit{tx}_{\textit{ID}}^{\prime}\right)\right);
        \end{align*}\\
        Generate a signature: $\sigma' = g^{(1/\textit{sk}_{p}'+h)}.$\\
        \textbf{return} $r'$ and signature $\sigma'$.
    \end{algorithm}

    \end{itemize}

    \item[4)] \textbf{Transaction Verification:} \Cref{fig:Transaction Verification} shows the \textsf{Epass} transaction verification. Users can create two different types of blockchain transactions, as was previously stated. Therefore, we consider two types of verification at this stage: immutable and redactable transaction verification.
    
    \begin{figure}[htpb!]
        \centering
        \includegraphics[width=0.48\textwidth]{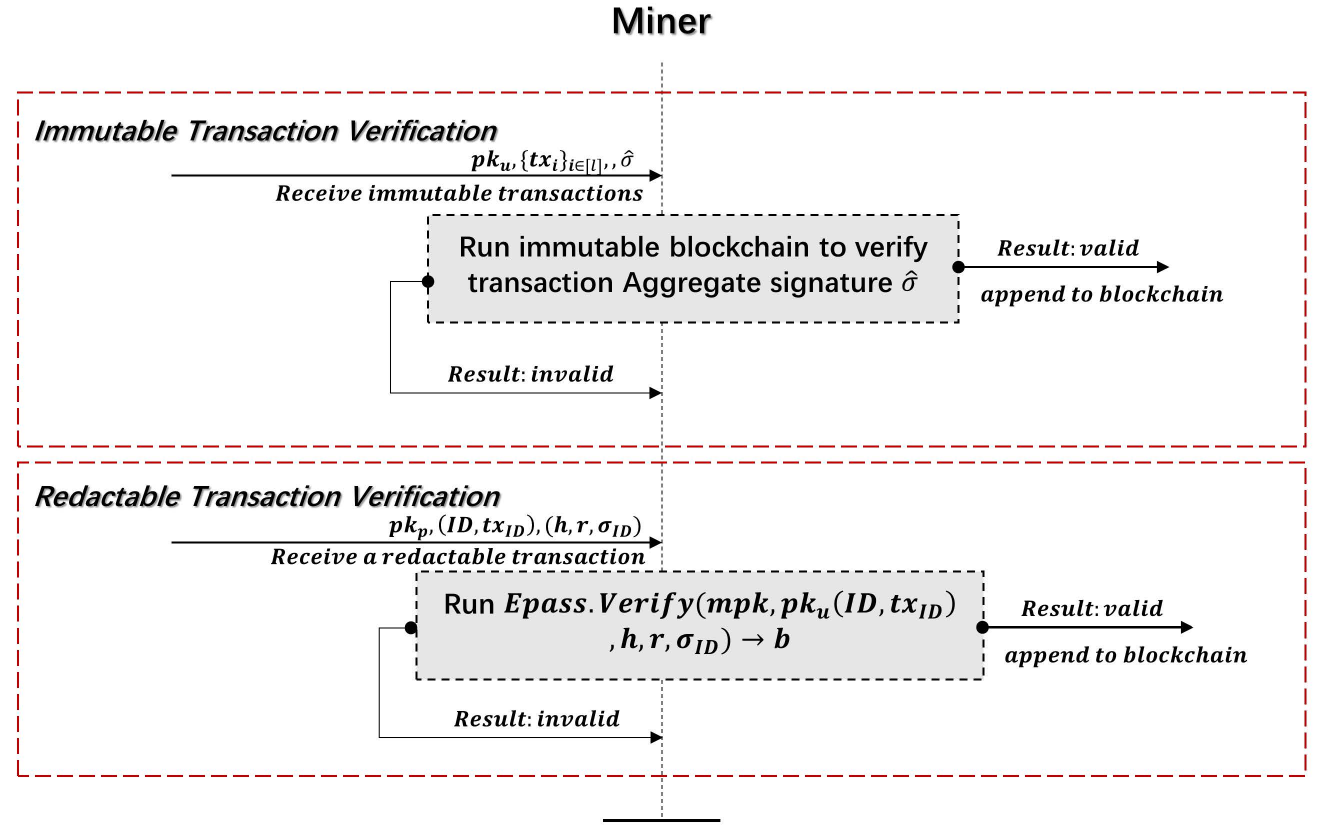}
        \caption{Transaction Verification}\label{fig:Transaction Verification}
    \end{figure}

    \begin{itemize}
        \item \textbf{Immutable transaction verification:} The miner receives a set of immutable transactions containing the user's public key $\textit{pk}_u$, the transaction set $ \{ \textit{tx}_{i}   \} _{i\in  [ l  ] }$, and the aggregated signature $\hat{\sigma}$. The immutable blockchain verification procedure is used by the miner to validate the transaction. If the validation proceeds correctly, the miner accepts the related transaction and adds it to the block; if not, the transaction is rejected.
        
        \item \textbf{Redactable transaction verification:} After receiving a redactable transaction containing the user's public key $pk_u$, transaction $ID, tx_{ID}$, and signature $h, r, \sigma_{\textit{ID}}$, the miner verifies the transaction by running $\textsf{Epass}.\textit{Verify}  ( \textit{mpk}, \textit{pk}_{u},  ( \textit{ID}, \textit{tx} _{\textit{ID}}  ), h, r, \sigma_{\textit{ID}}  )$ algorithm to verify the transaction, as shown in Algorithm 6. 
        The miner will update the local copy of the transaction if the validation is successful; otherwise, the miner will reject the transaction.
    \end{itemize}
    \begin{algorithm}\label{algr10}
        \caption{$\textsf{Epass}.\textit{Verify}  ( \textit{mpk}, \textit{pk}_{u}, \textit{pk}_p,  ( \textit{ID}, \textit{tx}_{\textit{ID}}  ),$\\$ h, r', \sigma', \widehat{\sigma}  )$}
        \KwIn{The transaction identity $\textit{ID}$ and its content $\textit{tx}_{\textit{ID}}$, the master public key $\textit{mpk}$, a public key $k_{u}$, a hash value $h$, a randomness $r$, and a signature $\sigma_{\textit{ID}} $.}
        \KwOut{A a decision bit $b\in \left \{ 0,1 \right \}$.}
        \eIf{the following conditions are true} 
        {$ \mathsf{CH.Verify}\left(\textit{pk}_{h},\left(\textit{ID}, \textit{tx}_{\textit{ID}}\right), h, r' \right)=1$\\
        $e(\widehat{\sigma}, \prod_{i=0}^{\ell}(g^{\alpha^{i}})^{\delta _{i}}) \stackrel{?}{=} e(g, g)$\\
        $e(\sigma', g^{i} g^{h'})\stackrel{?}{=}e(g, g)$\\
        {\textbf{return} $b=1$}}
        {\textbf{return} $b=0$}
    \end{algorithm}

\end{itemize}

\section{Security Analysis}\label{sa}
In this section, we demonstrate the security of the encryption algorithm and the unforgeability of the signature algorithm in \textsf {Epass}. Specifically, we analyze that \textsf{Epass} is \textbf{IND-CPA} and \textbf{EU-CMA} secure respectively.

\begin{theorem}
    \textsf{Epass} is \textbf{IND-CPA} secure, if the BDH assumption holds.
\end{theorem}

\begin{proof}
    We will prove that \textsf{Epass} is safe under the \textbf{IND-CPA} model by playing a game between the PPT adversary $\mathcal{A}$ and the simulator $\mathcal{S}$. 

    \noindent{\bf Setup}: $\mathcal{S}$ replaces the master secret key with $( \textit{pp}_\textsf{DS}, \alpha_0 )$, where $\alpha_0$ is chosen randomly from $\mathbb{Z}_{p}^{*}$ and is unknown to $\mathcal{S}$. Here, both $\alpha$ and $\alpha_0$ are random for $\mathcal{A}$, so $\mathcal{A}$ cannot distinguish between the real-world master secret key and the simulated one.

    \noindent{\bf Phase 1}: At any time $t\in T$, $\mathcal{A}$ can adaptively issue TIK extraction queries, and oracle will use $k_t$ as the response to each query.

    \noindent{\bf Challenge}: $\mathcal{A}$ sends $m_0, m_1 \in \mathcal{M}$ and $[t_{0}, t_{1}] \subseteq T$ to $\mathcal{S}$. $\mathcal{S}$ chooses a random bit $b$ and encrypts the message to get the challenge ciphertext $c'$. The simulation proceeds as follows: 
    \begin{itemize}
        \item [1.] $\mathcal{S}$ random choose $r_0^{*}$, and compute $r_0^{*}g$ and $r_0^{*}s\alpha'g$.
        \item [2.] $\mathcal{S}$ compute 
        \begin{align*} 
            K=\hat{e}\left(r_0^{*} s\alpha 'g, H_{1}(T)\right)=\hat{e}\left(g, H_{1}(T)\right)^{r_0^{*} s\alpha '}.
        \end{align*}
        \item [3.] $\mathcal{S}$ compute the ciphertext: 
        \begin{align*} 
            c'=<r_0^{*}g, sk_h \oplus H_{2}(K)>.
        \end{align*}
    \end{itemize}

    Finally, $\mathcal{S}$ sends the challenged ciphertext $c'$ to $\mathcal{A}$.

    \noindent{\bf Phase 2}: $\mathcal{A}$ continues to query the TIK extract oracle using the same restrictions as the Challenge phase.

    \noindent{\bf Guess}: $\mathcal{A}$ outputs its guess $b'$ for $b$. When $b' = b$, $\mathcal{A}$ successfully wins the game.

    Since $b'$ is a random guess of $\mathcal{A}$, there is no advantage to $\mathcal{S}$ from $\mathcal{A}$'s guesses, so we can obtain the advantage that $\mathcal{A}$ undermines the security of our proposed scheme, 
    \begin{align*} 
        \mathrm{Adv}_{\mathcal{A}} \leq \frac{q-1}{2 q} \epsilon_{\mathrm{BDH}}.
    \end{align*}

\end{proof}

\begin{theorem}
    \textsf{Epass} is \textbf{EU-CMA} secure, which has the following advantage: $  \operatorname{Adv}_{\mathcal{A}, \mathcal{\textsf{Epass}}}^{\textbf{EU-CMA}}(1^{\lambda})=   \operatorname{Adv}_{\mathcal{A}, \textsf{DS}}^{\textbf{EU-CMA}}(1^{\lambda}) $ , if the underlying digital signature $\textbf{DS}$ applied is \textbf{EU-CMA} secure.
\end{theorem}

\begin{proof}
    The standard security concept of our signature scheme is existential unforgeability under the choice message attack (\textbf{EU-CMA}) \cite{goldwasser1988digital}, which means that even if access to the signature oracle is gained, it is difficult to output a valid signature for a message $m$ that has never been requested to the signature oracle. \textsf{Epass} is \textbf{EU-CMA} secure if no probabilistic polynomial-time adversary $\mathcal{A}$ can win with non-negligible probability. We construct simulator $\mathcal{B}$, which has a non-negligible probability to break the underlying signature scheme $\mathcal{C}$. The security game is defined as follows.

    \noindent{\bf Setup}: 
    
    $\mathcal{B}$ sends $\lambda$ to $\mathcal{C}$ after receiving the security parameter $\lambda \in \mathbb{N}$, and $\mathcal{C}$ returns a public parameter $\textit{pp}_\textsf{DS}$. $\mathcal{B}$ sends $\lambda$ to $\mathcal{C}$ after receiving the security parameter $\lambda \in \mathbb{N}$, and $\mathcal{C}$ returns a public parameter $\textit{pp}_\textsf{DS}$. $\mathcal{B}$ initializes the public parameter $\textit{pp}_\textsf{CH} \gets \textsf{CH}.\textit{Setup} \left ( 1^{\lambda } \right )$, samples the public parameter of the message hash $hk\gets HGen\left (1^{\lambda} \right)$, and $\mathcal{C}$ gives $\mathcal{B}$ the public key $\textit{pk}_{\alpha }$. Next, $\mathcal{B}$ sets the public parameters $\textit{pp} = (\textit{pp}_\textsf{DS}, \textit{pp}_\textsf{CH}, \textit{hk})$ and the master public key $\textit{mpk} = (\textit{pk}_{\alpha }, H)$. Finally, $\mathcal{B}$ returns $(\textit{pp}, \textit{mpk})$ to $\mathcal{A}$ and gets lists:
    \begin{align*} 
        \mathcal{L}_{\mathrm{KG}_{\mathrm{p}}}, \mathcal{L}_{\textit{corr}_{\mathrm{p}}}, \mathcal{L}_{\mathrm{KG}_{\mathrm{u}}}, \mathcal{L}_{\textit{corr}_{\mathrm{u}}}, \mathcal{L}_{\mathrm{KG}_{\mathrm{S}}}, \mathcal{L}_{\textit{corr}_{\mathrm{S}}}, \mathcal{L}_{\mathrm{h}}, \mathcal{L}_{\textit{agg}}, \mathcal{L}_{\textit{apt}} \leftarrow \emptyset.
    \end{align*}

    \noindent{\bf Queries}: 
    
    In this phase, $\mathcal{A}$ adaptively queries the following oracle.

    \noindent$\mathcal{O}_{\textit{KeyGen}_{p}}(\textit{sk}_{p}, \textit{pk}_{p})$: $\mathcal{A}$ is allowed to query the provider key generation oracle. $\mathcal{C}$ provides $\{ g^{{\mu} ^{i} } \}_{i\in [B]}$ for $\mathcal{B}$. The algorithm returns a secret key $sk_p = ( \textit{pp}_\textsf{DS}, \textit{hk}, \cdot)$ and a public key $\textit{pk}_p = ( \textit{pp}_\textsf{DS}, \textit{hk}, \{ g^{{\mu} ^{i} } \}_{i\in  [B]})$. $\mathcal{B}$ updates the list 
    \begin{align*} 
    \mathcal{L}_{\mathrm{KG}_{\mathrm{p}}} \leftarrow \mathcal{L}_{\mathrm{KG}_{\mathrm{p}}} \cup\{(\textit{sk}_{p}, \textit{pk}_{p})\}
    \end{align*}
    and returns $\textit{pk}_p$ to $\mathcal{A}$.

    \noindent$\mathcal{O}_{\textit{Corrupt}_{p}}(\textit{pk}_p)$: $\mathcal{A}$ can query the provider corrupt oracle using the message on $\textit{pk}_p$. $\mathcal{B}$ sends the public key $\textit{pk}_p$ to $\mathcal{C}$ and gets the corresponding private key $\textit{sk}_p$. After that, $\mathcal{B}$ updates the list 
    \begin{align*}
    \mathcal{L}_{\textit{corr}_{\mathrm{p}}} \leftarrow \mathcal{L}_{\textit{corr}_{\mathrm{p}}} \cup\{\textit{pk}_p\}
    \end{align*}
    and returns $\textit{sk}_p$ to $\mathcal{A}$.

    \noindent$\mathcal{O}_{\textit{KeyGen}_{u}}(\textit{pk}_u)$: $\mathcal{A}$ is allowed to query the user key generation oracle. $\mathcal{B}$ receives $\textit{pk}_{\beta} $, $\textit{pk}_{h}$ and $\textit{pk}_{\textit{tre}}$ from $\mathcal{C}$ and immediately generates a local verification key $\textit{pk}_{\textit{local}} = \left(\Pi, \textit{hk}, g^{\alpha}\right)$. The algorithm returns a private key $\textit{sk}_u = (\cdot, \cdot)$ and a public key $\textit{pk}_u = \left (\textit{pk}_{\beta}, \textit{pk}_{h}, \textit{pk}_{\textit{tre}}, \textit{pk}_{\textit{local}} \right )$. Finally, $\mathcal{B}$ updates the list 
    \begin{align*}
    \mathcal{L}_{\textit{KG}_{\mathrm{u}}} \leftarrow \mathcal{L}_{\textit{KG}_{\mathrm{u}}} \cup\{\textit{pk}_u\}
    \end{align*}
    and returns $\textit{pk}_u$ to $\mathcal{A}$.

    \noindent$\mathcal{O}_{\textit{Corrupt}_{u}}(\textit{pk}_u)$: $\mathcal{A}$ can query the user corrupt oracle using the message on $\textit{pk}_u$. $\mathcal{B}$ sends $\textit{pk}_u$ to $\mathcal{C}$ and receives the corresponding $\textit{sk}_u$. $\mathcal{B}$ updates list
    \begin{align*}
    \mathcal{L}_{\textit{corr}_{\mathrm{u}}} \leftarrow \mathcal{L}_{\textit{corr}_{\mathrm{u}}} \cup\{\textit{pk}_u\}
    \end{align*}
    and returns $\textit{sk}_u$ to $\mathcal{A}$.

    \noindent$\mathcal{O}_{\textit{KeyGen}_{S}}(\textit{pk}_S)$: $\mathcal{A}$ is allowed to query the server key generation oracle. $\mathcal{B}$ receives $\textit{pk}_{S}$ from $\mathcal{C}$. The algorithm returns a private key $\textit{sk}_{S} = (\cdot)$ and a public key $\textit{pk}_{S}$. Finally, $\mathcal{B}$ updates the list 
    \begin{align*}
    \mathcal{L}_{\textit{KG}_{\mathrm{S}}} \leftarrow \mathcal{L}_{\textit{KG}_{\mathrm{S}}} \cup\{\textit{pk}_S\}
    \end{align*}
    and returns $\textit{pk}_S$ to $\mathcal{A}$.

    \noindent$\mathcal{O}_{\textit{Corrupt}_{S}}(\textit{pk}_S)$: $\mathcal{A}$ can query the server corrupt oracle using the message on $\textit{pk}_S$. $\mathcal{B}$ sends $\textit{pk}_S$ to $\mathcal{C}$ and receives the corresponding $\textit{sk}_S$. $\mathcal{B}$ updates list
    \begin{align*}
    \mathcal{L}_{\textit{corr}_{\mathrm{S}}} \leftarrow \mathcal{L}_{\textit{corr}_{\mathrm{S}}} \cup\{\textit{pk}_S\}
    \end{align*}
    and returns $\textit{sk}_S$ to $\mathcal{A}$.

    \noindent$\mathcal{O}_{\textit{TrCreat}}(\left(\textit{pk}_{u},\left(\textit{ID}, \textit{tx}_{\textit{ID}}\right), h, r, \sigma_{\textit{ID}}\right))$: $\mathcal{A}$ can query the hash oracle using the message on $\textit{pk}_u$, transaction $\textit{tx}_{\textit{ID}}$ and its corresponding identity $\textit{ID}$, hash value $h$, random number $r$ and signature $\sigma_{\textit{ID}}$. $\mathcal{B}$ selects a random number $r$ and computes the hash $h\gets \textsf{CH}.\textit{Hash}\left ( \textit{pk}_{h}, (\textit{ID}, \textit{tx}_{\textit{ID}}), r \right )$. Next, if $\textit{pk}_{u} \in \mathcal{L}_{\textit{corr}_{\mathrm{u}}}$, a signature $\sigma _{\textit{ID}} = g^{(1/\textit{sk}_{u}+h)}$ is generated; otherwise, $\mathcal{B}$ sends the public key $\textit{pk}_{u}$ and the message $(\textit{ID},r)$ to $\mathcal{C}$, and then gets $\sigma _{\textit{ID}}$ returned by $\mathcal{C}$. Finally, $\mathcal{B}$ updates the list 
    \begin{align*}
    \mathcal{L}_{\mathrm{h}} \leftarrow \mathcal{L}_{\mathrm{h}} \cup\left\{\left(\textit{pk}_{u},\left(\textit{ID}, \textit{tx}_{\textit{ID}}\right), h, r, \sigma_{\textit{ID}}\right)\right\}
    \end{align*}
    and returns $(h,r, \sigma _{\textit{ID}})$ to $\mathcal{A}$.

    \noindent$\mathcal{O}_{\textit{Aggregate}}( \textit{pk}_{u},  \{ ( ( \textit{ID}_{i} , r_{i}  ), \sigma_{i} ) \} _{i} )$: $\mathcal{A}$ can query the aggregate oracle using the message on $\textit{pk}_u$, and a transaction set $\{ ( ( \textit{ID}_{i} , r_{i}  ), \sigma_{i} ) \} _{i}$. $\mathcal{B}$ compute $\gamma _{i}=\frac{1}{\prod_{i \neq j}\left(x_{i}-x_{j}\right)}$, and $\hat{\sigma}=\prod_{i} \sigma_{i}^{\gamma i}$ if $p k_{u} \in \mathcal{L}_{\textit{corr}_{\mathrm{u}}}$, otherwise, $\mathcal{B}$ sends the public key $\textit{pk}_{u}$ and $\{ ( ( \textit{ID}_{i} , r_{i}  ), \sigma_{i} ) \} _{i}$ to $\mathcal{C}$, and then gets $\hat{\sigma }$ returned by $\mathcal{C}$. Finally, $\mathcal{B}$ updates the list
    \begin{align*}
        \mathcal{L}_{\mathrm{agg}} \leftarrow \mathcal{L}_{\mathrm{agg}} \cup\{\textit{pk}_{u},  \{ ( ( \textit{ID}_{i} , r_{i}  ), \sigma_{i} ) \} _{i}\}
    \end{align*}
    and returns $ \hat{\sigma }$ to $\mathcal{A}$.

    \noindent$\mathcal{O}_{\textit{Adapt}}( \textit{pk}_{p},\left(\textit{ID}, \textit{tx}_{\textit{ID}}\right), h, r, \sigma_{\textit{ID}},\left(\textit{ID}, \textit{tx}_{\textit{ID}}^{\prime}\right), r^{\prime}, \sigma_{\textit{ID}}^{\prime})$: $\mathcal{A}$ can query the hash oracle using the message on $\textit{pk}_u$, transaction $\textit{tx}_{\textit{ID}}$ and its corresponding identity $\textit{ID}$, hash value $h$, random number $r$, signature $\sigma_{\textit{ID}}$, transaction $\textit{tx}_{\textit{ID}}'$ and its corresponding identity $\textit{ID}$, new random number $r^{\prime}$, new signature $\sigma_{\textit{ID}}^{\prime}$. $\mathcal{B}$ generate a randomness $r^{\prime} \leftarrow \textsf{CH}.\textit{Adapt}(\textit{sk}_{h},(\textit{ID}, \textit{tx}_{\textit{ID}}), h, r,(\textit{ID}, \textit{tx}_{\textit{ID}}^{\prime}))$ and pick an index $i$ that is never used before to generate a signature $\sigma' = g^{(1/\textit{sk}_{p}^{'}+h)}$ if $\textit{sk}_p$ has been corrupted, such that $\textit{pk}_{p} \in \mathcal{L}_{\textit{corr}_{\mathrm{p}}}$; otherwise, $\mathcal{B}$ receives a signature $\sigma'$ form $\mathcal{C}$ after sending $\textit{pk}_{p}'$ and $(\textit{ID}, r')$ to $\mathcal{C}$. Finally, $\mathcal{B}$ updates the list 
    \begin{align*}
        \mathcal{L}_{\textit{apt}} \leftarrow \mathcal{L}_{\textit{apt}} \cup\left\{\left(\textit{pk}_{p},\left(\textit{ID}, \textit{tx}_{\textit{ID}}\right), h, r, \sigma_{\textit{ID}},\left(\textit{ID}, \textit{tx}_{\textit{ID}}^{\prime}\right), r^{\prime}, \sigma_{\textit{ID}}^{\prime}\right)\right\}
    \end{align*}
    and returns $r'$, $\sigma_{\textit{ID}}^{\prime}$ to $\mathcal{A}$.

    \noindent{\bf Output}: 

    $\mathcal{A}$ can be passed a tuple 
    \begin{align*}
    \left(pk^{*},(\textit{ID}^{*}, t x_{\textit{ID}^{*}}), h^{*}, r^{*}, \sigma_{\textit{ID}^{*}}, \widehat{\sigma}^{*}\right)
    \end{align*}
    as input to $\textsf{Epass}.\textit{Verify} (\textit{pk}^{*}, ( \textit{ID}, \textit{tx}_{\textit{ID}}  ), h^{*}, r^{*}, \sigma_{ \textit{ID}^{*}} ,\widehat{\sigma}^{*} )$.

\begin{itemize}
        \item $\mathcal{A}$ outputs a transaction that satisfies $\textit{pk}^{*}=\textit{pk}_{u}^{*} \wedge \textit{pk}_{u}^{*} \notin \mathcal{L}_{\textit{corr}_{u}} \wedge\left(\textit{pk}_{u}^{*},\left(\textit{ID}^{*}, \textit{tx}_{\textit{ID}^{*}}\right), h^{*}, r^{*}, \sigma_{\textit{ID}^{*}}\right) \notin \mathcal{L}_{\mathrm{h}} $ and is defined as a valid transaction, i.e., the transaction is uncorrupted and unqueried. Meanwhile, $\mathcal{A}$ can forge a signature $\sigma_{\textit{ID}}^{*}$ using $\textit{sk}_p$ to sign $(\textit{ID}^{*}, r^{*})$. And $\mathcal{B}$ breaks the signature scheme by sending the key $\textit{pk}_u$, the signature $\sigma_{\textit{ID}}^{*}$ and the message $(\textit{ID}^{*}, r^{*})$ to $\mathcal{C}$.
        \item $\mathcal{A}$ outputs a message unrelated to the adaption oracle and a redacted transaction, i.e. 
        \begin{align*}
            \textit{pk}^{*} &=\left(\textit{pk}_{u}^{*}, \textit{pk}_{p}^{*}\right) \wedge\left(\textit{pk}_{p}^{*}, \cdot, h^{*}, \cdot, \cdot,\left(\textit{ID}^{*}, \textit{tx}_{\textit{ID}^{*}}\right),r^{*}, \sigma_{\textit{ID}}\right),  
        \end{align*}
    and thus 
    \begin{align*}
         \textit{pk}^{*} &\notin \mathcal{L}_{\textit{apt}}. 
    \end{align*}
    Meanwhile, $\mathcal{A}$ can forge a signature $\sigma'$ using $\textit{sk}_m'$ to sign $(\textit{ID}^{*}, r^{*})$. And $\mathcal{B}$ breaks the signature scheme by sending the key $\textit{pk}_m'$, the signature $\sigma'$ and the message $(\textit{ID}^{*}, r^{*})$ to $\mathcal{C}$.
\end{itemize}
Therefore, through the above simulation process we can get the advantage 

\begin{align*}
    \operatorname{Adv}_{\mathcal{A}, \textsf{ Epass }}^{\textbf{EU-CMA}}\left(1^{\lambda}\right)=\operatorname{Adv}_{\mathcal{A}, \mathcal{D S}}^{\textbf{EU-CMA}}\left(1^{\lambda}\right).
\end{align*}

\end{proof}

    
    
    
    
    
    
    
    
\begin{table*}[t]
    \centering
    \caption{Experimental performances.}
    \label{table1}
    \begin{threeparttable}
    \begin{tabular}{|c|c|c|c|c|c|}
    \hline
    \multicolumn{2}{|c|}{\multirow{2}{*}{Algorithms}} & \multicolumn{4}{c|}{Time Cost in each algorithm(ms)\tnote{1} \tnote{2} ($k$ = 8 / 16 / 24)}\\
    \cline{3-6}
    
    \multicolumn{2}{|c|}{} & $u$ = 8 & $u$ = 16 & $u$ = 24 & $u$ = 32 \\
    \hline
    \multicolumn{2}{|c|}{Syetem Setup} & 149 / 156 / 161 & 151 / 163 / 164 & 162 / 171 / 174 & 164 / 174 / 176 \\ 
    \hline
    \multicolumn{2}{|c|}{User Key Generation} & 440 / 463 / 494 & 438 / 460 / 495 & 439 / 461 / 496 & 440 / 464 / 495 \\ 
    \hline
    \multicolumn{2}{|c|}{Provider Key Generation} & 15 / 14 / 14 & 14 / 13 / 15 & 15 / 16 / 18 & 15 / 17 / 16 \\ 
    \hline
    \multicolumn{2}{|c|}{Server Key Generation} & 32 / 30 / 32 & 29 / 30 / 31 & 29 / 29 / 28 & 31 / 30 / 31 \\ 
    \hline
    \end{tabular}
    \label{table_MAP}
    \begin{tablenotes}
        \footnotesize
        \item[1] $k$ represents the number of deferred transactions for users.
        \item[2] $u$ represents the number of user.
    \end{tablenotes}
\end{threeparttable}
\end{table*}

\begin{figure*}[tb]
    \centering
    
    \subfigure[\label{fig:a}]{
    \includegraphics[width=5.6cm]{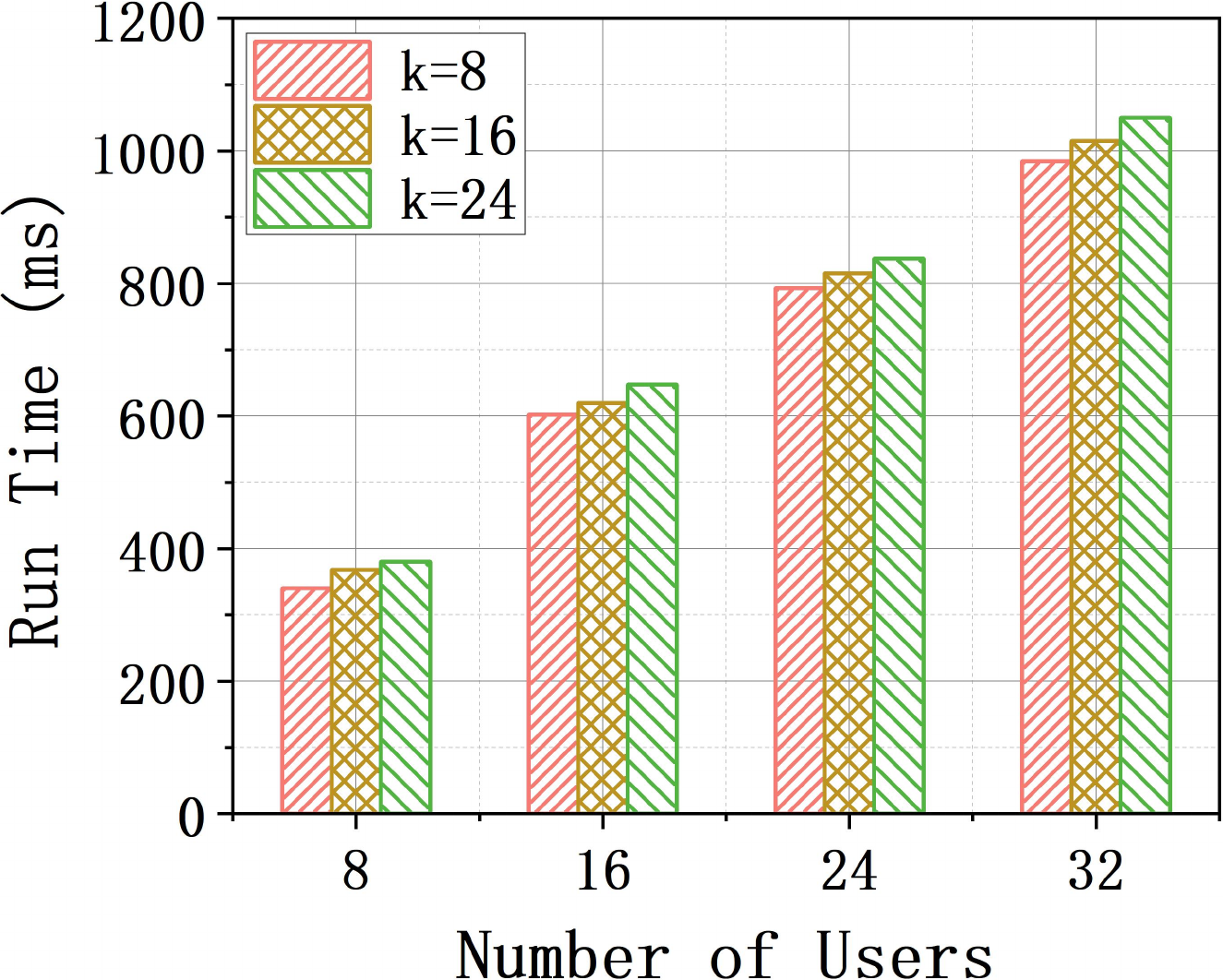}
    }
    \subfigure[\label{fig:b}]{
    \includegraphics[width=5.6cm]{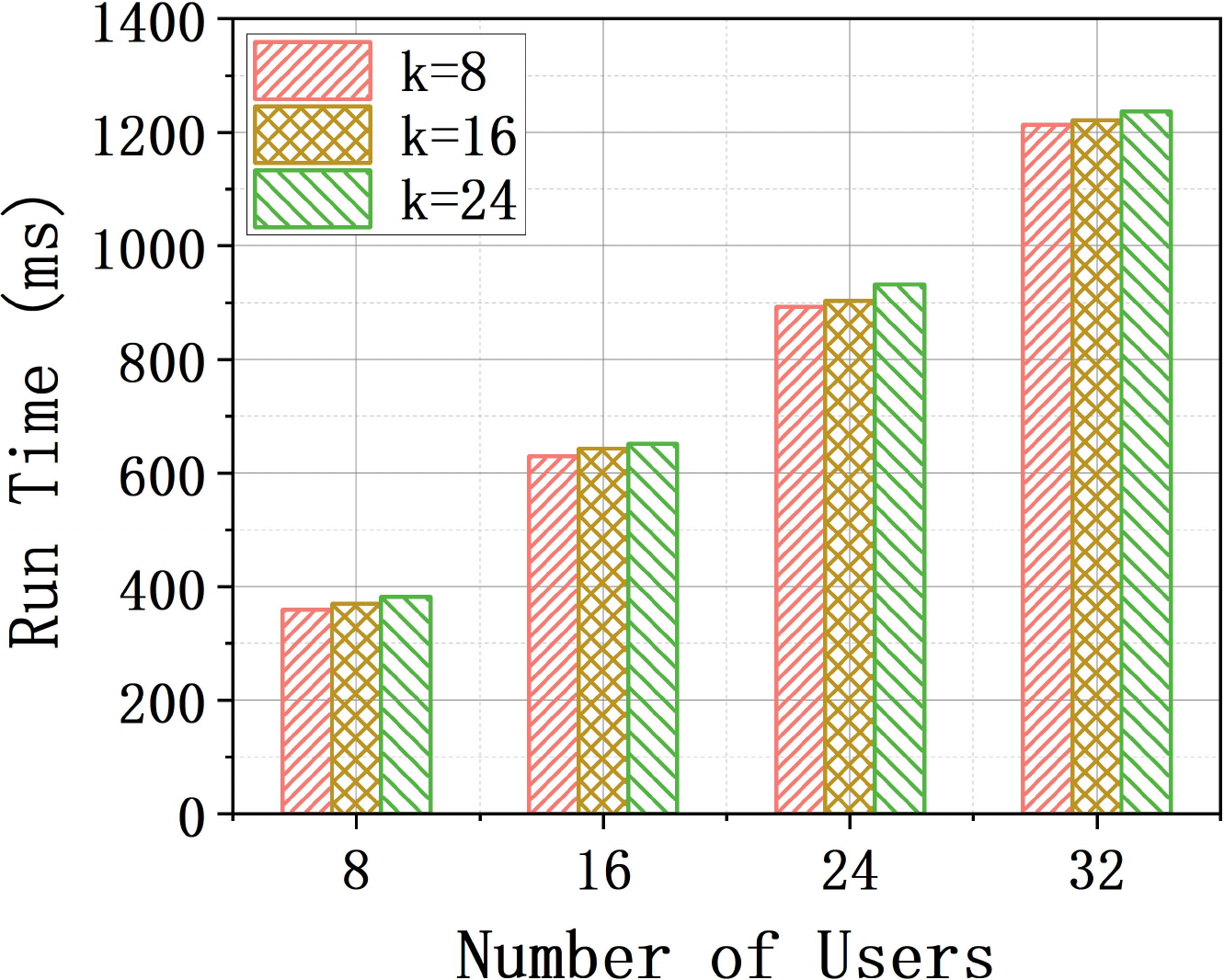}
    
    }
    \subfigure[\label{fig:c}]{
    \includegraphics[width=5.6cm]{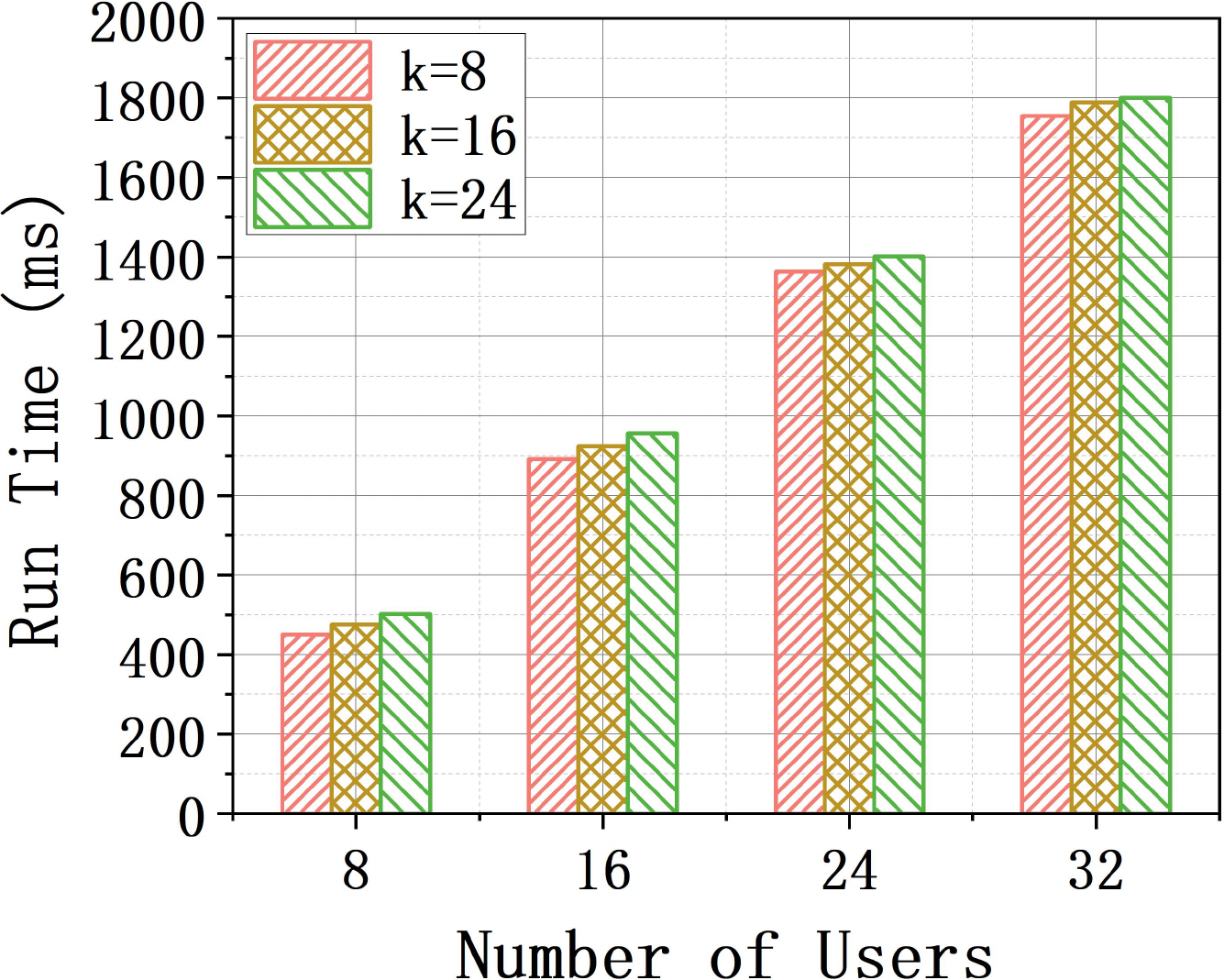}
    
    }
    
    \quad
    \subfigure[\label{fig:d}]{
    \includegraphics[width=5.6cm]{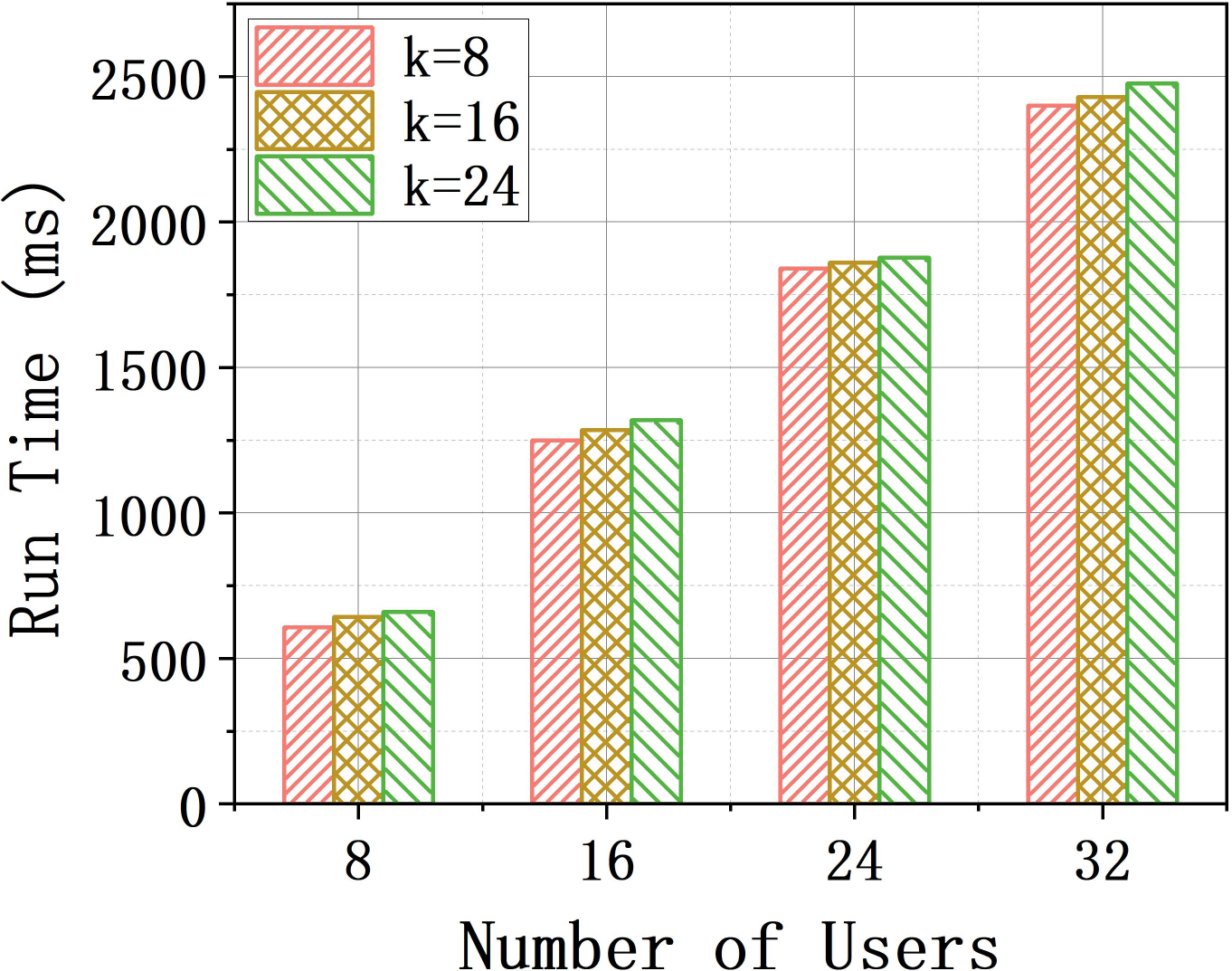}
    
    }
    \subfigure[\label{fig:e}]{
    \includegraphics[width=5.6cm]{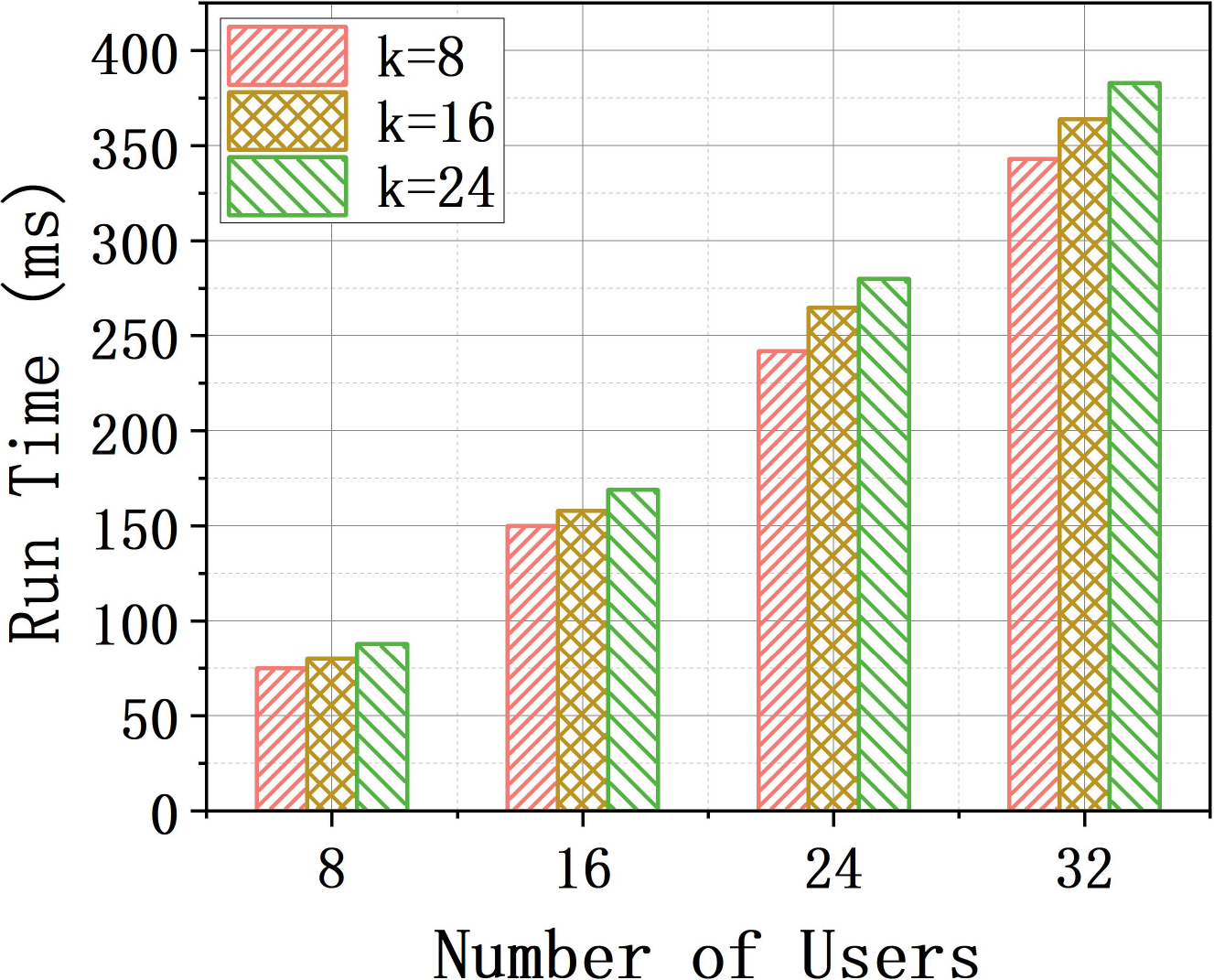}
    
    }
    \subfigure[\label{fig:f}]{
    \includegraphics[width=5.6cm]{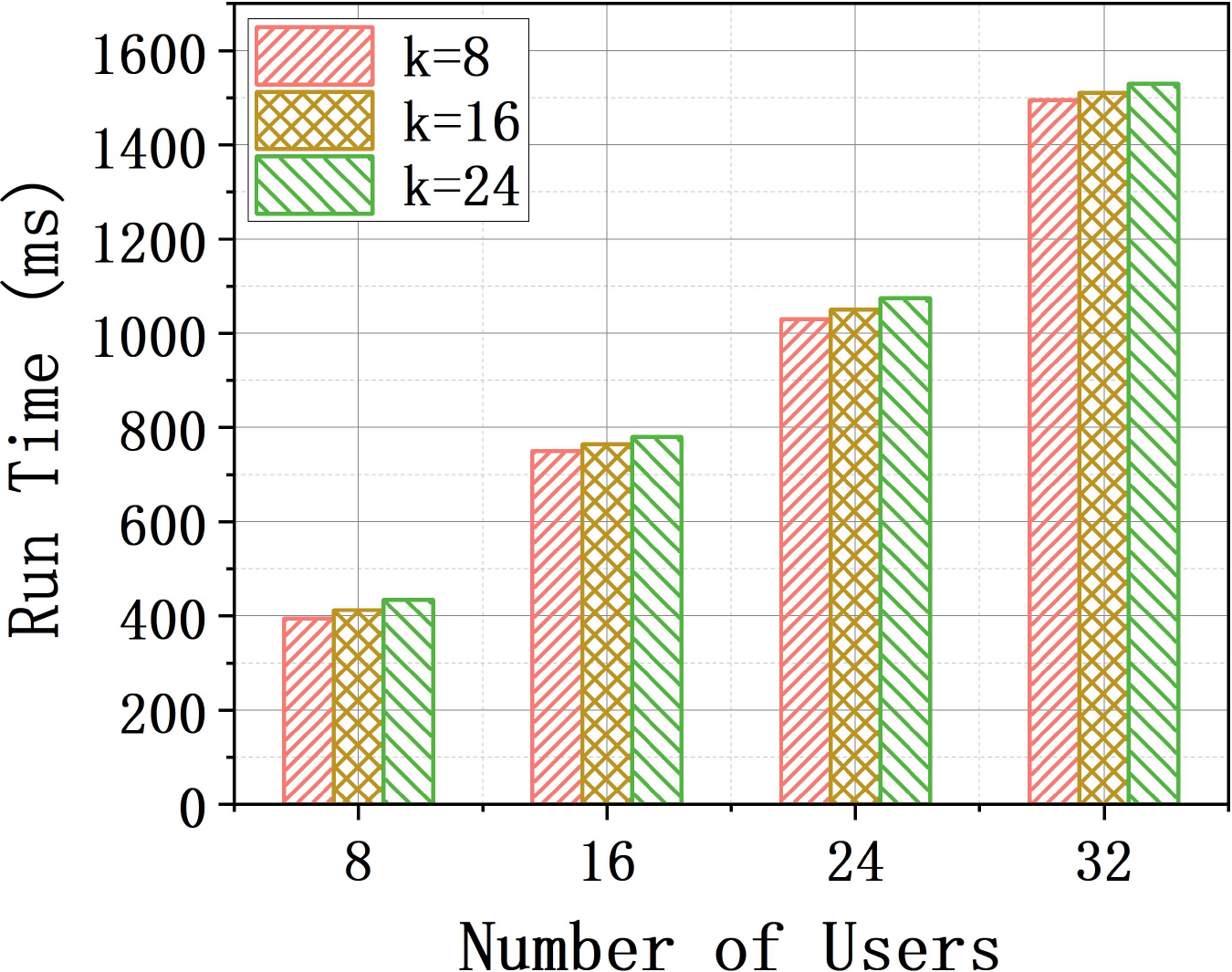}
    
    }
    
    \caption{Experimental performances. (a) time for transaction creation; (b) time for signature aggregation; (c) time for extraction; (d) time for timed-release decryption; (e) time for adaption; (f) time for verification. $k$ represents the number of deferred transactions for users.
}
    \label{fig:Experiment}
\end{figure*}

\section{Performance}\label{performance} 
In this section, we evaluate and analyze the performance of \textsf{Epass}.

\subsection{Settings}
\textbf{Setup.} The operating environment we conducted our experiments was Windows 11, the system type was a 64-bit operating system, and the device was configured with an Intel(R) Core(TM) i5-10210U CPU @ 1.60 GHz 2.11 GHz and 8 GB of RAM on board. We used Java to implement the simulation, using the Java PairingBased Cryptography Library (JPBC) and Type A pairings to perform the \textsf{Epass} scheme. In the measurement, we use Java's function \textbf{Java.lang.System.currentTimeMillis()} to measure the running time from the start of the operation to the end of the operation. The experimental results are shown in \Cref{table1}, where $k$ represents the number of deferred transactions for users.

\textbf{Baselines for Comparison.} 
We use the locally verifiable signature algorithm proposed in \cite{goyal2022locally} (\textsf{LVS}) as a baseline. \textsf{LVS} can protect users' deferred payment information from being exploited by malicious third parties and reduce the time overhead of the system. However, \textsf{LVS} requires separate validation for each deferred payment transaction, which does not satisfy the efficiency of our design goals. Therefore, \textsf{Epass} extends the single verification in \textsf{LVS} to a subset of validations to reduce the time overhead.

Next, the most advanced implementation of BNPL schemes \cite{Paylaterproject2022, Atpay, Apenow} handles deferred payment transactions via Ethernet smart contracts, so we use \textsf{Geth} to represent Ethernet smart contracts. Note that \textsf{Geth} implements deferred payments by generating new transactions, which somewhat undermines the scalability of the BNPL service. For this reason, \textsf{Epass} rewrites deferred payment transactions through the trapdoor of the chameleon hash \cite{krawczyk1998chameleon}, which does not generate additional new transactions during deferred payment and reduces the burden that \textsf{Geth} imposes on the blockchain.

\textbf{Indicators.} In the experiment, we evaluated the following indicators: \textrm{i) \textit{time costs:}} the time cost of system initialization, transaction making, transaction rewriting, and transaction verification in \textsf{Epass}, and \textrm{ii) \textit{communication costs:}} the size of the ciphertext generated by the user in the transaction-making phase.

\begin{figure*}[tb]
    \centering
    \subfigure[\label{fig:b}]{
    \includegraphics[width=5.6cm]{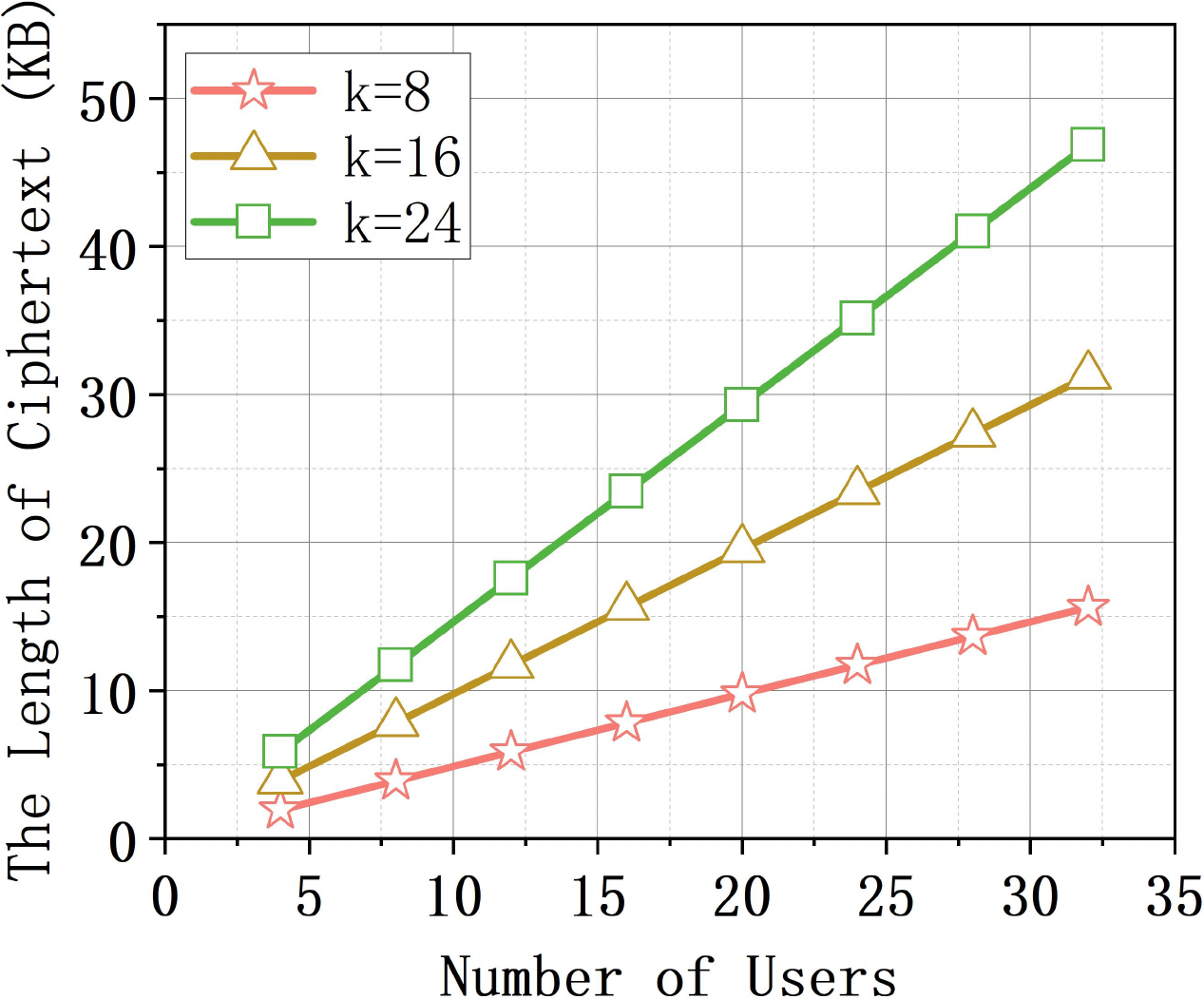}
    \label{fig:communication}
    }
    \subfigure[\label{fig:c}]{
    \includegraphics[width=5.8cm]{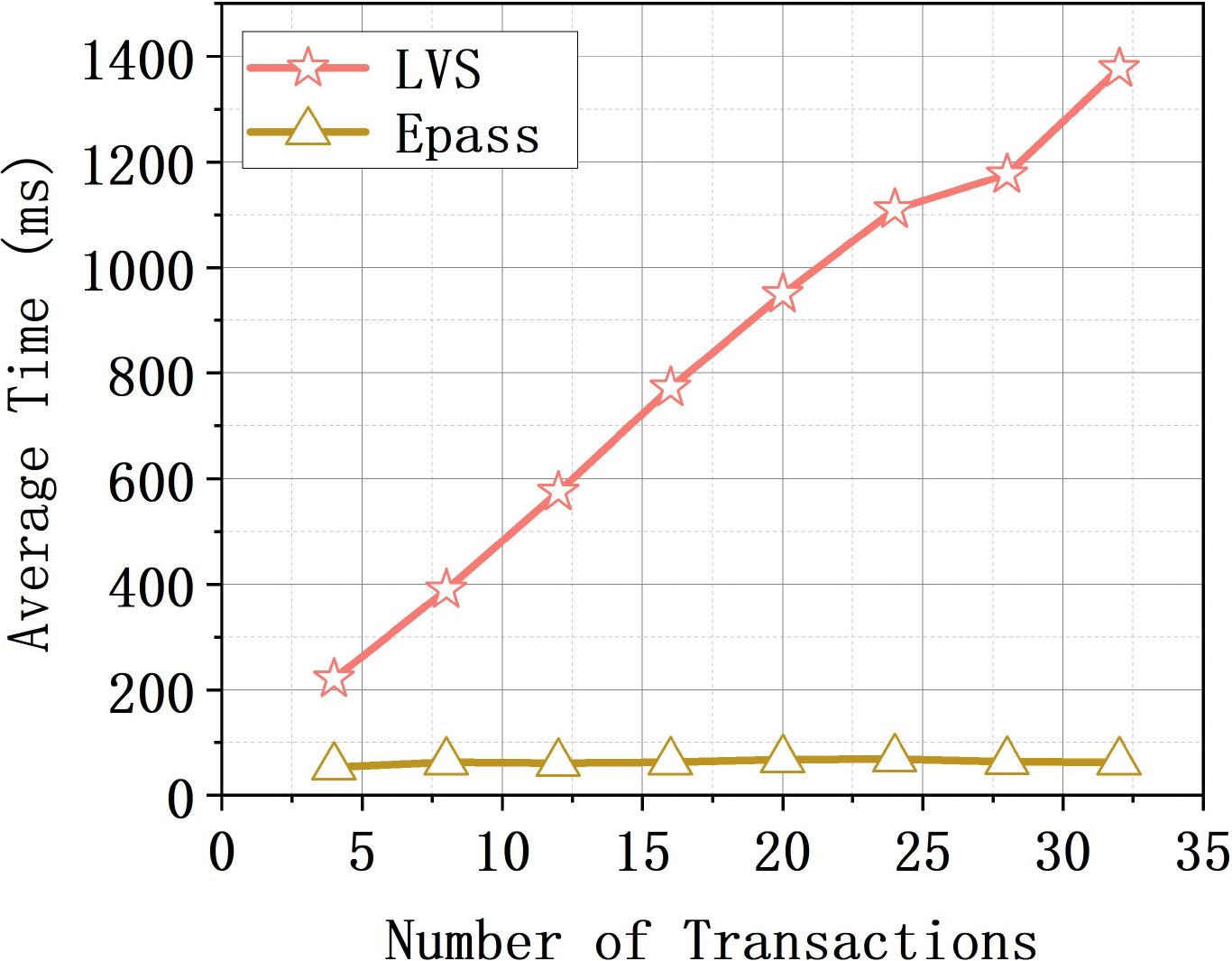}
    \label{fig:12}
    }
    \subfigure[\label{fig:a}]{
    \includegraphics[width=5.3cm]{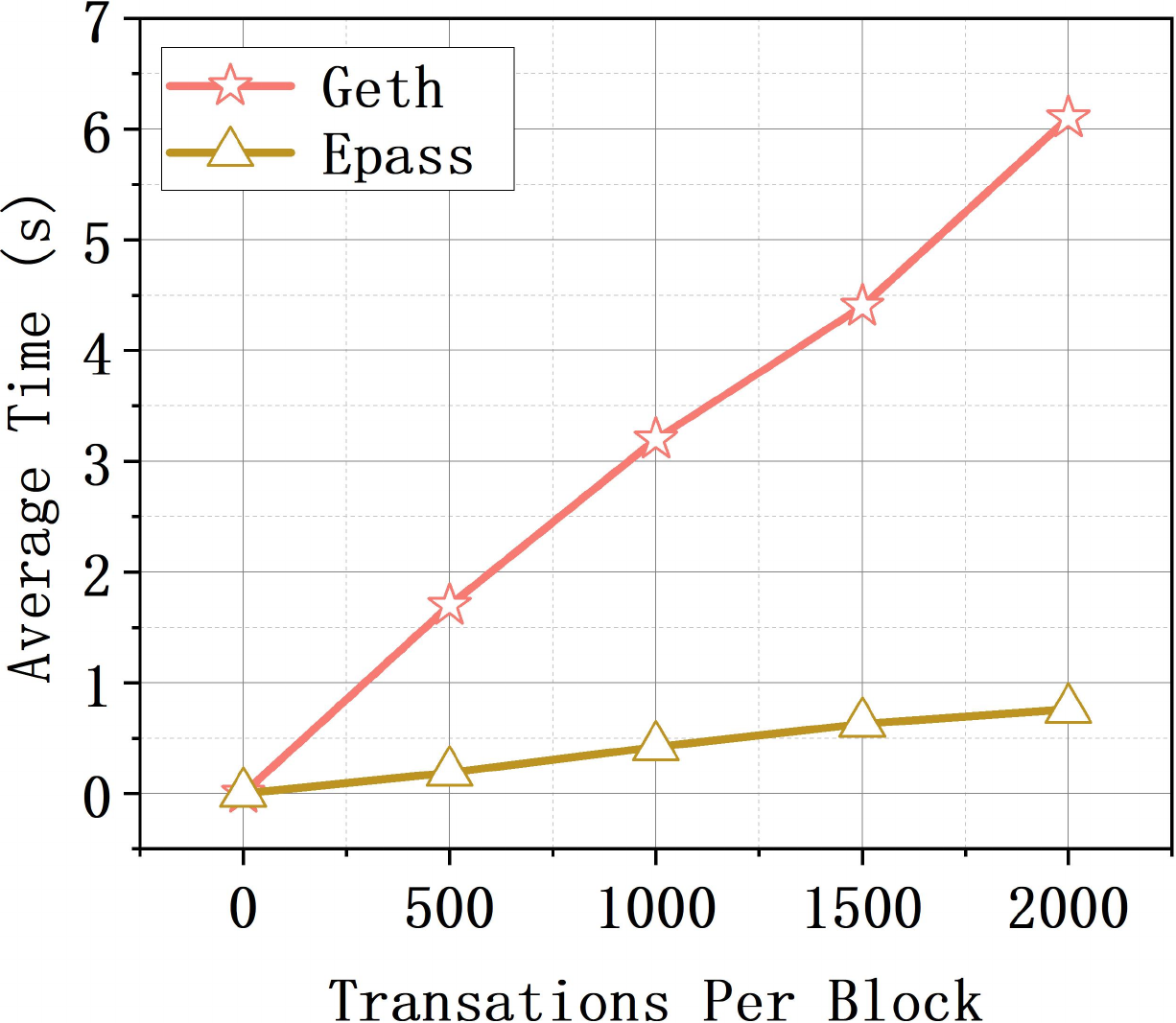}
    \label{fig:11}
    }

    \caption{Experimental performances. (a) time for process deferred transactions; (b) communication costs; (c) time for locally verification.
}
\end{figure*}

\subsection{Evaluation Results}

\textbf{Time costs.} As shown in \Cref{table1}, the time consumption in the system initialization phase (system setup, user key generation, provider key generation, and server key generation) is stable as the number of users and the number of deferred transactions increases. Among them, the primary time cost of the system initialization phase is concentrated on user key generation and system setup. With the number of users being 32 and the number of deferred transactions being 24, the time consumption of user key generation is less than 500ms, while the time consumption of system setup is less than 180ms. The time consumption of provider key generation and server key generation is negligible compared to the previous two phases. As the number of users and the number of deferred transactions increases, the time consumption of server key generation remains between 25 and 35 ms, while the time consumption of provider key generation is consistently below 20 ms. The results show that the time cost of \textsf{Epass} is acceptable during the system initialization phase. 

The time consumption of transaction making (transaction creation, signature aggregation, and extraction), transaction rewriting (timed-release decryption and adaption), and transaction verification (verification) phases tend to increase linearly with the increase of input dimension. As shown in \Cref{fig:Experiment}, the input dimension appears as a set of users and a set of deferred payment transactions, and the size of the set depends on the number of users and deferred payment transactions, where $k$ denotes the number of deferred payment transactions. The experimental results show that the primary time consumption comes from the transaction making phase. With the number of users at 32 and the number of deferred transactions at 24, the time consumption of the transaction making phase reaches 4000ms, of which transaction creation consumes about 1050ms, signature aggregation consumes about 1240ms, and extraction consumes about 1800ms. In the same input dimension, the time consumption of the transaction rewriting phase and transaction verification phase is 2860ms and 1500ms, respectively. In contrast, the time consumption of timed-release decryption, adoption, and verification is 2500ms, 380ms, and 1500ms, respectively. Moreover, the number of deferred transactions has less impact on the time consumption compared to the number of users. This is due to the fact that the operations involved in delayed transactions are usually exponential in nature and do not impose a significant burden on the computational overhead. In summary, the time cost of \textsf{Epass} is acceptable.

\textbf{Communication costs.} As shown in \Cref{fig:communication}, we use the size of the ciphertext generated by users in the transaction making phase to evaluate the communication cost for the different numbers of users and the different numbers of deferred transactions. \Cref{fig:communication} shows that the communication costs increase linearly with the number of users and the number of deferred transactions. When the number of users reaches 32, and the number of deferred transactions reaches 24, the communication cost is still less than 50KB. The experimental results show that the number of users and the number of deferred transactions impact the communication cost. The communication cost of \textsf{Epass} is acceptable when appropriate parameters are chosen.

\textbf{Comparison with baseline.} Next, we extend \textsf{LVS} even further with the locally verifiable signature that supports subset verification. As shown in \Cref{fig:11}, the time consumption of \textsf{LVS} grows linearly with the number of signatures to be verified. As the number of signatures to be verified increases from 4 to 32, the time consumption grows from 200ms to 1400ms. In comparison, when the number of signatures to be verified is 4 and 32, the time consumption of \textsf{Epass} to verify these signatures is less than 50ms and 70ms, respectively. The results show that the time cost of \textsf{Epass} is acceptable when dealing with local verification of signatures.

Finally, we created a private test network using Go-Ethereum and compared the time required to process deferred transactions between the baseline and \textsf{Epass}. We performed block creation experiments in which the block difficulty target was set to $0x0400$. This allowed us to accurately measure the operation time without increasing the work permit overhead while avoiding causing extreme blockchain forks. As shown in \Cref{fig:12}, the time overheads of \textsf{Epass} and \textsf{Geth} both grow linearly. When transactions per block reach 2000, the time overhead of \textsf{Epass} is less than 1s, while \textsf{Geth} reaches 6s, which is six times higher than our scheme. The experimental results show that as the number of transactions in each block increases, \textsf{Epass} leads to better efficiency and is more suitable for real-world applications.

\section{Related Work}  \label{overview}
This section briefly reviews the background and related work on BNPL.

\subsection{Background}
Over the past few years, BNPL has gradually increased in popularity among financial institutions, merchants, and consumers as online shopping has proliferated in pandemic proportions \cite{sengupta2022bnpl, muparadzi2021business}. BNPL companies have created one of the fastest-growing segments of consumer finance, according to GlobalData \cite{feng2023pricing}. BNPL allows consumers to purchase a product immediately and pay for it over a period of time, usually in fixed installments \cite{guttman2022buy}. It is a reverse-assist system where the purchaser can immediately get the product or service and then pay for it. Like how digital assets and blockchain technology have infiltrated most businesses, BNPL platforms enable their customers to use cryptocurrencies. Affirm \cite{Affirm}, Zip \cite{Zip}, Klarna \cite{Klarna}, and XRPayNet \cite{XRPayNet} are experimenting with their blockchains to create an application that offers BNPL functionality. Affirm's app brings all its products together in one place. At the same time, it's google chrome extension allows customers to use Affirm at various retailers, even if the service is not integrated into their checkout options \cite{Affirm}. Zip offers the same type of service as Affirm, allowing customers to buy now and then pay weekly or monthly, which is great flexibility for customers. In addition, Zip offers a digital wallet that can carry up to \$1,000 and is interest-free \cite{fisher2021developments}. Klarna works with more than 5,000 banks in 18 European countries by partnering with Swedish cryptocurrency broker Safello \cite{Safello} to provide it with open banking infrastructure. Safello's 180,000 customers will now use Klarna's available banking payment system to buy cryptocurrencies without leaving Safello's platform \cite{Klarna&Safello}. XRPayNet allows businesses to continue using their existing processing systems, making the crypto-to-fiat payment process seamless \cite{XRPayNet2022}. As more and more people discover cryptocurrency, businesses are beginning to accept it as a payment method. It was only a matter of time before platforms combined it with BNPL options.

\subsection{Blockchain-based Buy Now Pay Pater}
The explosive growth of the e-commerce industry has provided consumers worldwide with multiple ways to purchase their favorite products while saving time and money. With the integration of options such as cryptocurrency and BNPL, customers and merchants alike are reaping additional benefits. However, only some e-commerce platforms offer these benefits to customers. But so far, there has been little discussion about blockchain-based BNPL. Until recently, only a few popular companies were offering this service through smart contracts, including PLP \cite{Paylaterproject2022}, Atpay \cite{Atpay}, and Apenow \cite{Apenow}. There are two types of participants in these solutions, the user (consumer) and the merchant. Users on the blockchain-based BNPL platform can make purchases and earn rewards. Verified users can shop online and in-store from any platform-approved merchant. At checkout, users can select the BNPL feature to pay. When they complete their repayment, they receive rewards in the form of tokens that can also be used for future purchases or to earn higher purchase limits. Merchants on the blockchain-based BNPL platform can grow and expand their business and incentivize their customers with marketing campaigns. Using the blockchain-based BNPL service, users and merchants can benefit from it.

PLP \cite{Paylaterproject2022} is a DEFI protocol and is the first BNPL platform built to integrate blockchain technology with its own cryptocurrency. While providing significant cost savings to all participants in the ecosystem by leveraging smart contract technology and blockchain, PLP will also allow users to pay for their purchases with any recognized cryptocurrency they hold in their PLP wallet. 

Atpay \cite{Atpay} is merging blockchain and cryptocurrency technology with the BNPL concept, enabling consumers to shop online and offline and use specific cryptocurrencies when paying from the platform's native wallets. Customers get access to a wide range of payment options, significant cost savings, and rewards when they shop. The platform is also supported by its native cryptocurrency @Pay tokens, allowing holders to participate in the management of the agreement. 

Apenow \cite{Apenow} is an NFT installment purchase agreement built on Teller that supports the ability of buyers to finance their next purchase. Users can make a down payment on an NFT purchase, and the remaining amount can be paid in installments. The platform holds the NFT in escrow until the full payment is completed, and the user can withdraw the NFT to their wallet only after the payment is completed. If the user does not complete the installments by the agreed deadline, the NFT will be liquidated and reimbursed to the loan provider. 

However, these solutions do not consider the privacy issues arising from on-chain data transparency and the additional time overhead caused by deferred payments. Specifically, the transactions generated by users are visible to all nodes in the blockchain so that a malicious third party may analyze users' wealth status based on their deferred payment transaction information. Besides, deferred payment transactions generated by users incur additional time overhead, which undermines the scalability of the blockchain-based BNPL service. Therefore, in this work, we construct an efficient and privacy-preserving blockchain deferred payment solution to address the problems that exist in the current work.

\section{Conclusion}\label{conl}
This work investigates the privacy and efficiency issues associated with the BNPL model in the blockchain. To address these issues, we propose \textsf{Epass}, an efficient and privacy-preserving blockchain-based asynchronous payment scheme. By extending single verification of locally verifiable signatures to a subset of verification, \textsf{Epass} has lower time consumption while protecting consumer privacy. Asynchronous payments are implemented by leveraging timed-release encryption and redactable blockchain, saving time and arithmetic power compared to existing schemes. Extensive experiments and security analysis show that \textsf{Epass} has practical security features and acceptable communication cost (KB-level). The time overhead was reduced by more than four times compared to the baseline.



\bibliographystyle{IEEEtran}
\bibliography{reference}

\begin{IEEEbiography}[{\includegraphics[width=1in,height=1.25in,clip,keepaspectratio]{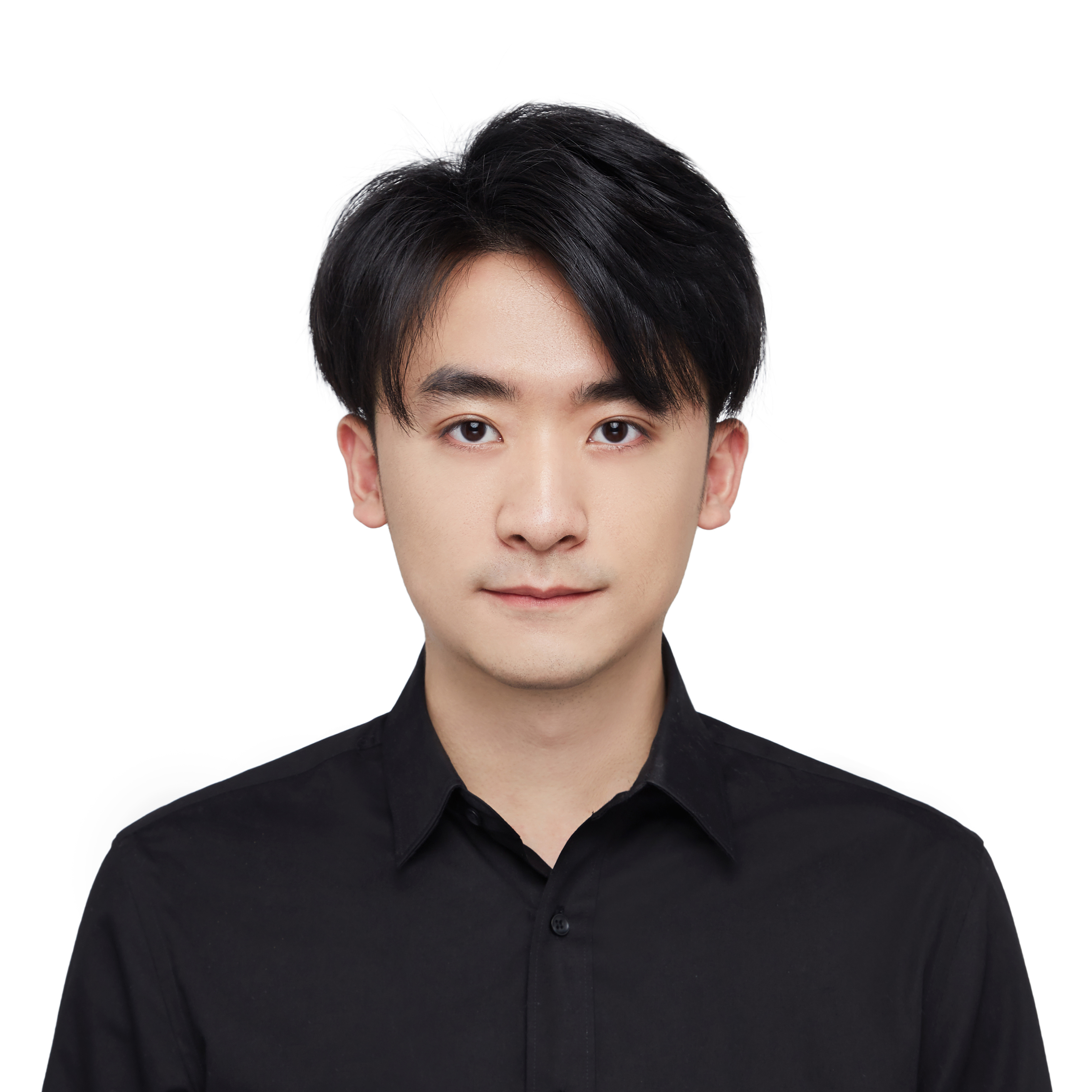}}]{Weijie Wang} received his B.S. degree from Xidian University in 2020. He is currently a master student in the School of Computer Science at Beijing Institute of Technology. His research interests include federal learning, security and privacy in blockchain.
    \end{IEEEbiography}

    \begin{IEEEbiography}[{\includegraphics[width=1in,height=1.25in, clip, keepaspectratio]{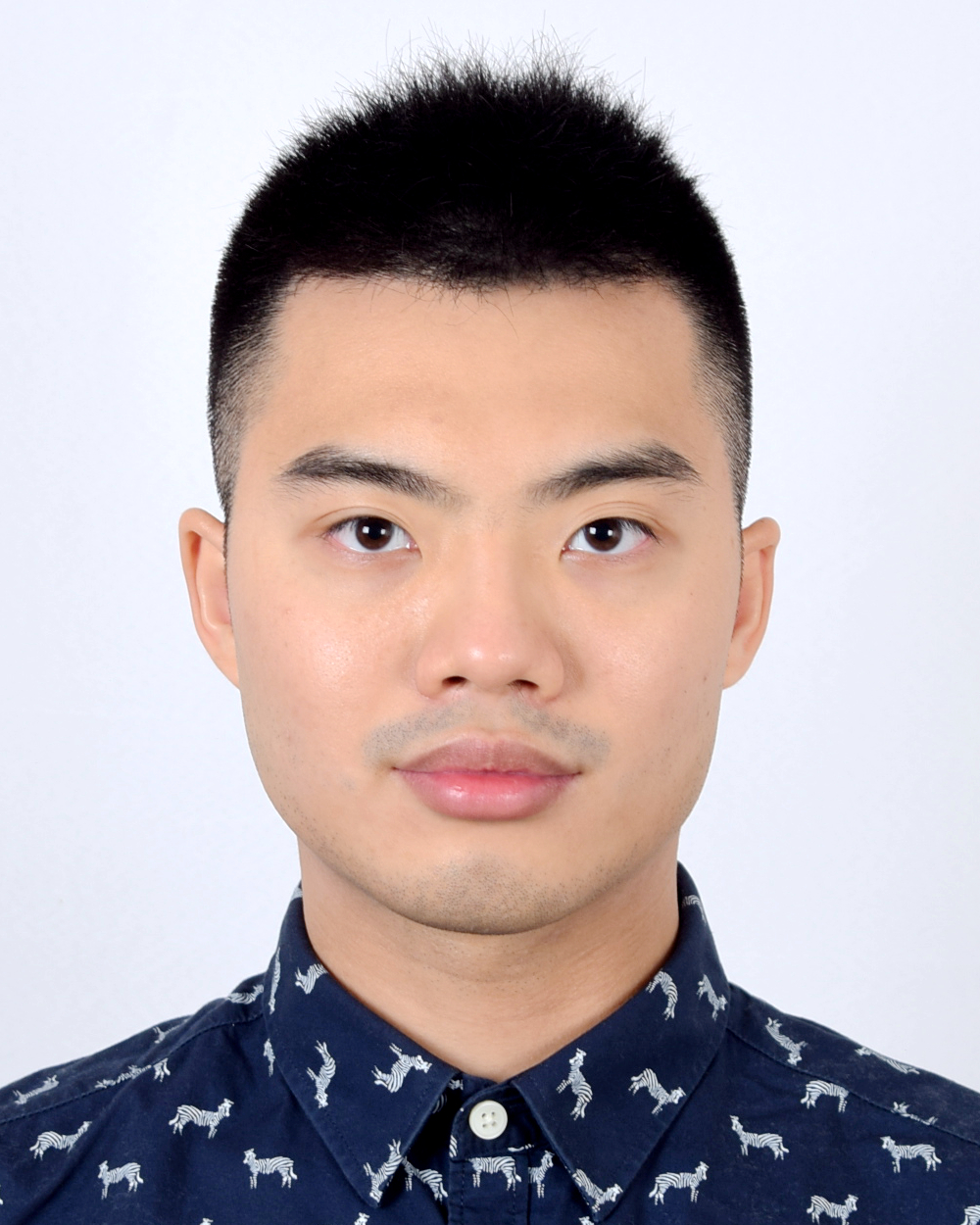}}]{Jinwen Liang} is a Postdoctoral fellow in the Department of Computing of the Hong Kong Polytechnic University. He received his Ph.D. degree and B.S. degree from Hunan University, China, in 2021 and 2015, respectively. From 2018 to 2020, he was a visiting Ph.D. student at BBCR Lab, University of Waterloo, Canada. His research interests include applied cryptography, AI security, blockchain, and database security. He served as the Technical Program Committee Chair of the 1st international workshop on Future Mobile Computing and Networking for Internet of Things (IEEE FMobile 2022), Publicity Chair of the 6th International Workshop on Cyberspace Security (IWCSS 2022), TPC Member of IEEE VTC' 19 Fall. He is a member of the IEEE. 
    \end{IEEEbiography}

    \begin{IEEEbiography}[{\includegraphics[width=1in,height=1.25in,clip,keepaspectratio]{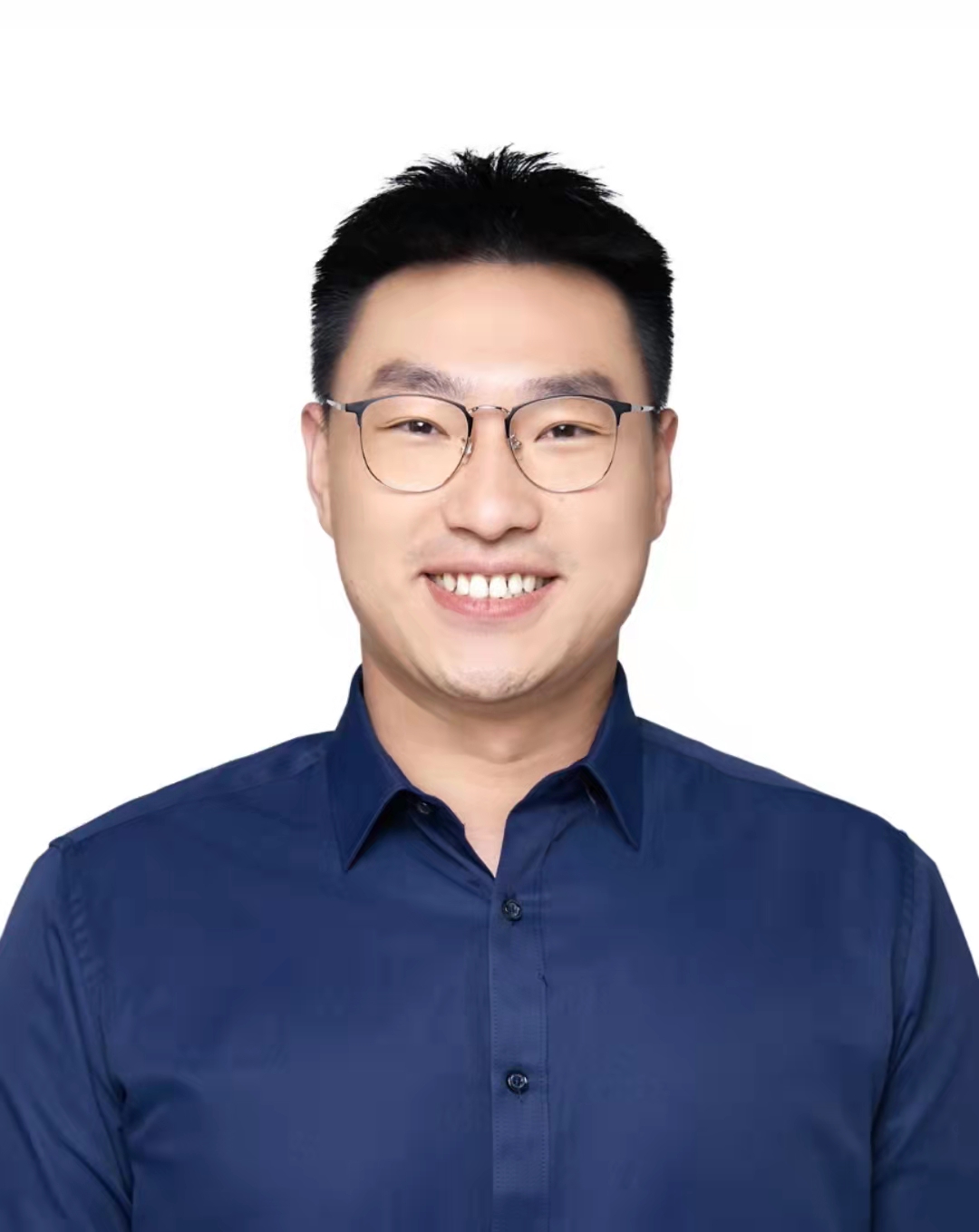}}]{Chuan Zhang} received his Ph.D. degree in computer science from Beijing Institute of Technology, Beijing, China, in 2021. From Sept. 2019 to Sept. 2020, he worked as a visiting Ph.D. student with the BBCR Group, Department of Electrical and Computer Engineering, University of Waterloo, Canada. He is currently an assistant professor at School of Cyberspace Science and Technology, Beijing Institute of Technology. His research interests include secure data services in cloud computing, applied cryptography, machine learning, and blockchain.
    \end{IEEEbiography}

    \begin{IEEEbiography}[{\includegraphics[width=1in,height=1.25in]{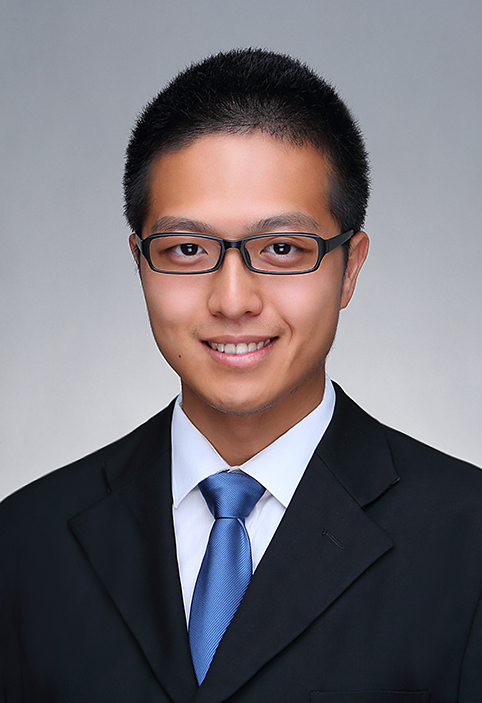}}]{Ximeng Li} (Senior Member, IEEE) received the B.S. degree in electronic engineering from Xidian University, Xi’an, China, in 2010 and the PhD degree in cryptography from Xidian University, China, in 2015. Currently, he is a full professor with the College of Mathematics and Computer Science, Fuzhou University, China. Also, he is a research fellow with the School of Information System, Singapore Management University, Singapore. His research interests include cloud security, applied cryptography and big data security.
    \end{IEEEbiography}

    \begin{IEEEbiography}[{\includegraphics[width=1in,height=1.25in]{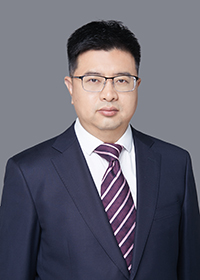}}]{Liehuang Zhu} (Senior Member, IEEE) received his Ph.D. degree in computer science from Beijing Institute of Technology, Beijing, China, in 2004. He is currently a professor at the School of Cyberspace Science and Technology, Beijing Institute of Technology. His research interests include security protocol analysis and design, group key exchange protocols, wireless sensor networks, and cloud computing.
    \end{IEEEbiography}

    \begin{IEEEbiography}[{\includegraphics[width=1in,height=1.25in]{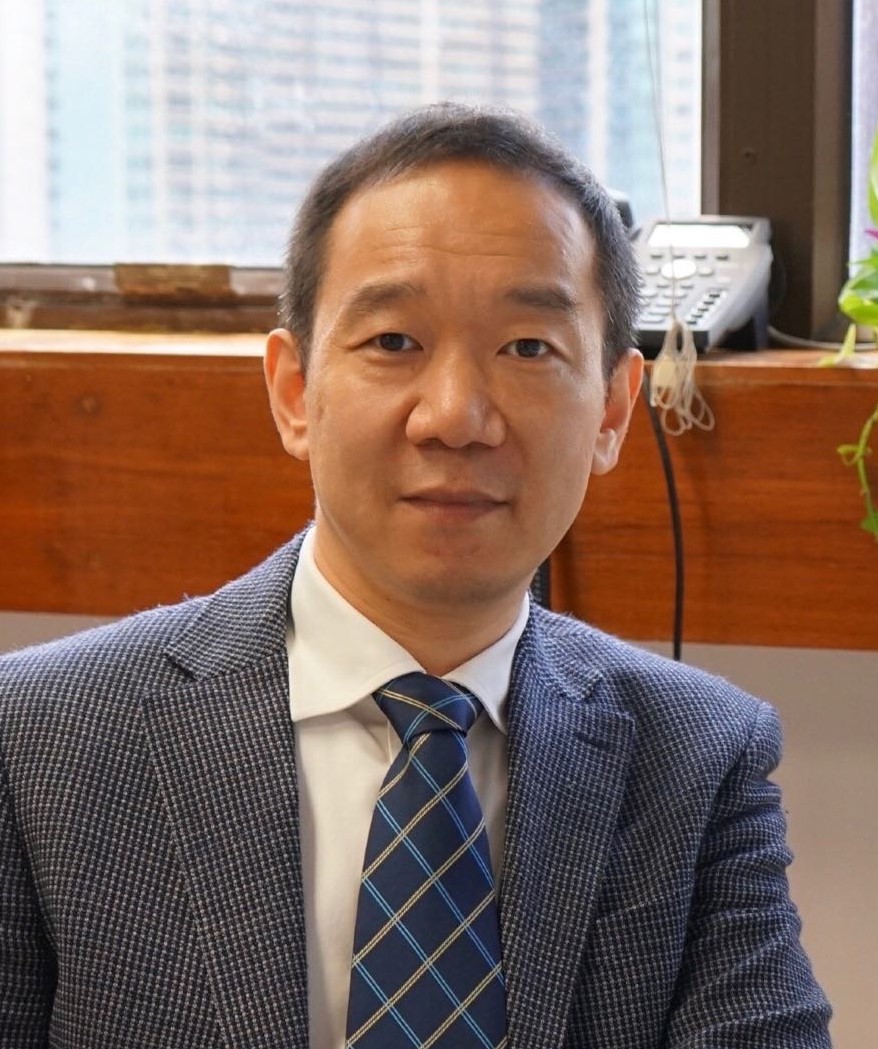}}]{Song Guo} (Fellow, IEEE) is a Full Professor in the Department of Computing at The Hong Kong Polytechnic University. He also holds a Changjiang Chair Professorship awarded by the Ministry of Education of China. His research interests are mainly in the areas of big data, edge AI, mobile computing, and distributed systems. With many impactful papers published in top venues in these areas, he has been recognized as a Highly Cited Researcher (Web of Science) and received over 12 Best Paper Awards from IEEE/ACM conferences, journals and technical committees. Prof. Guo is the Editor-in-Chief of IEEE Open Journal of the Computer Society. He has served on IEEE Communications Society Board of Governors, IEEE Computer Society Fellow Evaluation Committee, and editorial board of a number of prestigious international journals like IEEE Transactions on Parallel and Distributed Systems, IEEE Transactions on Cloud Computing, IEEE Internet of Things Journal, etc. He has also served as chair of organizing and technical committees of many international conferences. Prof. Guo is an IEEE Fellow and an ACM Distinguished Member.
    \end{IEEEbiography}

\vfill

\end{document}